\documentclass[letterpaper, 11pt]{article}
\usepackage{graphicx}
\usepackage{amsmath}
\usepackage{amsfonts}
\usepackage{hyperref}
\usepackage{amsthm}
 \usepackage{amscd}
\usepackage{mathrsfs}
\usepackage{caption}
\usepackage{subcaption}
\usepackage{float}
\usepackage[top=1in,left=1in,right=1in,bottom=1in]{geometry}

\newtheorem{theorem}{Theorem}[section]

\newtheorem{example}{Example}[section]
\counterwithin{figure}{section}
\hypersetup{
	colorlinks = true,
	linkcolor  = red,
	citecolor  = green,
	filecolor  = cyan,
	urlcolor   = magenta,}

\numberwithin{equation}{section}

\title{Inverse scattering method for an integrable system of derivative nonlinear Schr\"odinger equations}
\author{Mehmet Unlu\\
Department of Mathematics\\
Recep Tayyip Erdogan University\\
53100 Rize, Turkey
}

\date{}

\begin{document}

\maketitle

\begin{abstract}
We present a solution method for the integrable system of nonlinear partial differential equations, known as the DNSL II
system (derivative nonlinear Schr\"odinger II system) or the Chen--Lee--Liu system. This is done by presenting a solution
technique for the inverse scattering problem for the corresponding linear system of ordinary differential equations with
energy-dependent potentials. The relevant inverse scattering problem is solved by establishing a system of linear integral 
equations, which we refer to as the Marchenko system of linear integral equations. In solving the inverse scattering problem 
we use the input data set consisting of a transmission coefficient, a reflection coefficient, and the bound-state information 
presented in the form of a pair of matrix triplets. Using our data set as input to the Marchenko system, we recover the 
potentials from the solution to the Marchenko system. By using the time-evolved input data set, we recover the time-evolved 
potentials, where those potentials form a solution to the integrable DNLS II system. 
\end{abstract}



\newpage

\section{Introduction}
\label{section1}

In this paper we present a solution technique for the nonlinear system of partial differential equations
\begin{equation}
\label{1.1}
\begin{cases}
i q_t+q_{xx}-i q q_x r=0,
\\
\noalign{\medskip}
i r_t- r_{xx}-i q r r_x
=0,
\end{cases}
\end{equation}
where $x$ and $t$ are the independent variables taking values on the real axis $\mathbb R,$ the subscripts denote the 
respective partial derivatives, the dependent variables $q$ and $r$ are complex-valued functions of $x$ and $t.$
We refer to $x$ as the spacial coordinate and $t$ as the time coordinate. The nonlinear system \eqref{1.1} is known
\cite{AC1991,APT2003,AS1981,AEU2023a,AEU2023b,CLL1979,K1984,OS1998a,OS1998b,T2010,TW1999a,TW1999b}
as the DNLS II system (derivative nonlinear Schr\"odinger II system) or the Chen--Lee--Liu system. It has
\cite{AC1991,KN1978} important physical applications in propagation of electromagnetic waves in nonlinear media, 
propagation of hydromagnetic waves traveling in a magnetic field, and transmission of ultra short nonlinear pulses in optical fibers.  

The nonlinear system \eqref{1.1} is related to another nonlinear system given by
\begin{equation}
\label{1.2}
\begin{cases}
i \tilde q_t+\tilde q_{xx}-i(\tilde q^2 \tilde r)_x=0,
\\
\noalign{\medskip}
i\tilde r_t-\tilde r_{xx}-i(\tilde q \tilde r^2)_x=0,
\end{cases}
\end{equation}
which is known \cite{AC1991,AS1981,CS1989,F1963,F1967,GGKM1967,KN1978,L1968,L1987,M1986,N1980,N1983,NMPZ1984}
as the DNSL I system or the Kaup--Newell system. We use a tilde to denote the quantities related to \eqref{1.2}. We observe that the linear parts
in \eqref{1.1} and \eqref{1.2} coincide but their nonlinear parts differ from each other. Both nonlinear systems
\eqref{1.1} and \eqref{1.2} are known to be integrable in the sense of the inverse scattering transform 
\cite{A2009,GGKM1967,L1987}. The integrability of \eqref{1.1} is assured \cite{AC1991,AKNS1974,AS1981,NMPZ1984}
by the existence of the corresponding AKNS pair $(\mathcal X,\mathcal T),$ 
where $\mathcal X$ and $\mathcal T$ are the $2\times 2$ matrix-valued functions containing $q$ and $r,$ their partial $x$-derivatives, and the 
spectral parameter denoted by $\zeta.$ Similarly, the integrability of \eqref{1.2} is assured by the existence of the 
corresponding AKNS pair $(\tilde{\mathcal X},\tilde{\mathcal T}),$ where $\tilde{\mathcal X}$ and $\tilde{\mathcal X}$ are
the $2\times 2$ matrix-valued functions containing $\tilde q$ and $\tilde r,$ their partial $x$-derivatives, and the spectral parameter $\zeta.$ 

The matrix $\mathcal X$ appears in the first-order linear system of the differential equations given by
\begin{equation}\label{1.3}
\displaystyle\frac{d}{dx}\begin{bmatrix}
\alpha\\
\noalign{\medskip}
\beta
\end{bmatrix}
=\mathcal X \begin{bmatrix}
\alpha\\
\noalign{\medskip}
\beta
\end{bmatrix}, \qquad x\in\mathbb R,
\end{equation}
where the quantities $\alpha$ and $\beta$ are the two components of the wavefunction
depending on the independent variable $x.$ The independent variable $t$ in \eqref{1.1} appears in \eqref{1.3} as a 
parameter. The AKNS pair matrices $\mathcal X$ and $\mathcal T$ associated with \eqref{1.1} are given by
\begin{equation}
\label{1.4}
\mathcal X=
\begin{bmatrix}
-i\zeta^2& \zeta q
\\
\noalign{\medskip}
\zeta r&i\zeta^2+\displaystyle\frac{i}{2} q r
\end{bmatrix},
\end{equation}
\begin{equation}
\label{1.5}
\mathcal T=\begin{bmatrix}
-2i\zeta^4-i\zeta^2 q r&2\zeta^3 q+\zeta\left(i q_x+\displaystyle\frac{1}{2} q^2 r\right)
\\
\noalign{\medskip}
2\zeta^3r+\zeta\left(-i r_x+\displaystyle\frac{1}{2} q r^2\right)&2i\zeta^4+i\zeta^2 q r +\displaystyle\frac{1}{2}\,(q r_x-q_x r)+\displaystyle\frac{i }{4}q^2 r^2
\end{bmatrix}.
\end{equation}
Using \eqref{1.4} in \eqref{1.3}, we write the corresponding first-order linear system as
\begin{equation}\label{1.6}
\displaystyle\frac{d}{dx}\begin{bmatrix}
\alpha\\
\noalign{\medskip}
\beta
\end{bmatrix}=
\begin{bmatrix}
-i\zeta^2 & \zeta q\\
\noalign{\medskip}
\zeta r & i\zeta^2+\displaystyle\frac{i}{2} \,qr
\end{bmatrix}
\begin{bmatrix}
\alpha\\
\noalign{\medskip}
\beta
\end{bmatrix},\qquad x\in\mathbb R.
\end{equation}
The matrices $\mathcal X$ and $\mathcal T$ appearing in \eqref{1.4} and \eqref{1.5} form the AKNS pair for the nonlinear system \eqref{1.1} 
in the sense that the $2\times 2$ matrix-valued system of equations given by
\begin{equation}\label{1.7}
\mathcal X_t-\mathcal T_x+\mathcal X \mathcal T-\mathcal T \mathcal X=0,
\end{equation}
yields \eqref{1.1}. In other words, by using the matrices $\mathcal X$ and $\mathcal T,$ we evaluate the $2\times 2$ matrix appearing on the 
left-hand side of \eqref{1.7}, where each entry is a polynomial in $\zeta.$ Those polynomials identically vanish for all $x,$ $t,$
and $\zeta$ provided that $q$ and $r$ satisfy the nonlinear system \eqref{1.1}. 

In a similar manner, the $2\times 2$ AKNS pair matrices $\tilde{\mathcal X}$ and $\tilde{\mathcal T}$ associated with the nonlinear system
\eqref{1.2} are given by
\begin{equation*}
\tilde{\mathcal X}=\begin{bmatrix}
-i\zeta^2& \zeta \tilde q\\
\noalign{\medskip}
\zeta \tilde r& i\zeta^2
\end{bmatrix},
\end{equation*}
\begin{equation*}
\tilde{\mathcal T}=\begin{bmatrix}
-2i\zeta^4-i\zeta^2\tilde q\tilde r  
& 2\zeta^3\tilde q+\zeta(i\tilde q_x+\tilde q^2 \tilde r)\\
\noalign{\medskip}
2\zeta^3\tilde r+\zeta(-i\tilde r_x+\tilde q\tilde r^2)& 2i\zeta^4+i\zeta^2\tilde q\tilde r
\end{bmatrix}.
\end{equation*}
The matrix $\tilde{\mathcal X}$ appears in the first-order linear system of differential equations given by
\begin{equation}\label{1.10}
\displaystyle\frac{d}{dx}\begin{bmatrix}
\tilde\alpha\\
\noalign{\medskip}
\tilde\beta
\end{bmatrix}
=\tilde{\mathcal X} \begin{bmatrix}
\tilde\alpha\\
\noalign{\medskip}
\tilde\beta
\end{bmatrix},
\end{equation}
where $x$ is the independent variable, $t$ appears as a parameter, and the dependent variables
$\tilde\alpha$ and $\tilde\beta$ are functions $x$ and they also contain the parameter $t.$ We write the linear system \eqref{1.10} explicitly as
\begin{equation}\label{1.11}
\displaystyle\frac{d}{dx}\begin{bmatrix}
\tilde\alpha\\
\noalign{\medskip}
\tilde\beta
\end{bmatrix}=
\begin{bmatrix}
-i\zeta^2& \zeta \tilde q
\\
\noalign{\medskip}
\zeta \tilde r&i\zeta^2
\end{bmatrix}
\begin{bmatrix}
\tilde\alpha\\
\noalign{\medskip}
\tilde\beta
\end{bmatrix},\qquad x\in\mathbb R.
\end{equation}
Because of the appearance of the spectral parameter $\zeta$ in $\mathcal X$ and $\tilde{\mathcal X},$ the linear systems \eqref{1.6} and
\eqref{1.11} also depend on the spectral parameter $\zeta.$ 
We refer to $q$ and $r$ appearing in \eqref{1.6} as the potentials. Since $q$ and $r$ appear in \eqref{1.6} in the form of
$\zeta q$ and $\zeta r,$ we refer to the potentials in \eqref{1.6} as energy-dependent potentials. This is because the spectral
parameter $\zeta$ appearing in \eqref{1.6} has the physical interpretation as energy.

In our paper we relate the solutions to the linear systems \eqref{1.6} and \eqref{1.11} through the transformation
\begin{equation}
\label{1.12}
\begin{bmatrix}
\alpha\\
\noalign{\medskip}
\beta
\end{bmatrix}=c\begin{bmatrix}
1& 0\\
\noalign{\medskip}
0&E(x)
\end{bmatrix}\begin{bmatrix}
\tilde\alpha\\
\noalign{\medskip}
\tilde\beta
\end{bmatrix},
\end{equation}
where $c$ is an arbitrary constant and $E(x)$ is the complex-valued scalar quantity given as
\begin{equation}\label{1.13}
E(x):=\exp\left(\displaystyle\frac{i}{2}\displaystyle\int_{-\infty}^x dy\, q(y)\,r(y)\right).
\end{equation}
The arbitrary constant $c$ is determined by a spacial asymptotic condition when we use a particular solution.
Since the quantity $E(x)$ defined in \eqref{1.13} is in general complex valued, it
does not necessarily have the unit amplitude.  From \eqref{1.13} it follows that
\begin{equation}
\label{1.14}
\displaystyle\lim_{x\to-\infty}E(x)=1, \quad \displaystyle\lim_{x\to+\infty}E(x)=e^{i\mu/2},
\end{equation}
where we have defined the complex constant $\mu$ as
\begin{equation}
\label{1.15}
\mu:=\int_{-\infty}^\infty dy\,q(y)\,r(y).
\end{equation}

Our goal in this paper is to analyze the direct and inverse scattering problems for the linear system \eqref{1.6}
with the goal of developing a solution method for the nonlinear system \eqref{1.1}. In order to present our method in the 
simplest way, we assume that the potentials $q$ and $r$ in \eqref{1.6} belong to the Schwartz 
class in $x\in\mathbb R$ for each fixed $t.$ We recall that the Schwartz class $\mathcal S(\mathbb R)$ consists of functions 
of $x$ where the derivatives of all orders exist and are continuous and those derivatives vanish as $x\to\pm\infty$ faster than
any negative power of $|x|.$ Although our results hold under weaker conditions on the potentials, we present our ideas in the 
simplest form by assuming that $q$ and $r$ belong to $\mathcal S(\mathbb R).$ Since $t$ appears as a
parameter in the analysis of the direct and inverse scattering problems for \eqref{1.6}, we mostly suppress the
$t$-dependence for the related quantities. For example, we write $q(x)$ and $r(x)$ instead of $q(x,t)$ and $r(x,t)$ in
Sections~\ref{section1}-\ref{section4}. On the other hand, we display the argument $t$ in Sections~\ref{section5} and \ref{section6} for further clarity.

When the potentials $q$ and $r$ belong to the Schwartz class $\mathcal S(\mathbb R),$ we have a scattering scenario 
associated with \eqref{1.6}. In other words, the potentials $q$ and $r$ vanish as $x\to\pm\infty,$ and hence any solution to 
\eqref{1.6} for $\zeta\in\mathbb R$ behave as a linear combination of the column vectors $\begin{bmatrix}
e^{-i\zeta^2 x}\\0\end{bmatrix}$ and $\begin{bmatrix}0\\e^{i\zeta^2 x}\end{bmatrix}.$
As indicated in Section~\ref{section2}, this allows us to use some particular solutions to \eqref{1.6} known as the Jost
solutions. From the spacial asymptotics of the Jost solutions, we introduce the scattering coefficients that are functions of 
the spectral parameter $\zeta.$ For $\zeta\in\mathbb R,$ we know that \eqref{1.6} has only two linearly independent column-vector
solutions. The linear system \eqref{1.6} may also have column-vector solutions that are square integrable in $x\in\mathbb R.$ 
Such solutions are known as bound-state solutions to \eqref{1.6}, and they occur only at certain complex values of the 
spectral parameter $\zeta.$ At each bound-state $\zeta$-value, it is possible to have two or more linearly independent 
solutions to \eqref{1.6}. The number of such linearly independent solutions determines the multiplicity of the bound state at 
that $\zeta$-value. We refer to the bound-state $\zeta$-values, their multiplicities, and the normalization constants associated 
with each of the bound-state multiplicity collectively as the bound-state information.

The direct scattering problem \eqref{1.6} consists of the determination of the scattering coefficients and the bound-state 
information when the potentials $q$ and $r$ are known. On the other hand, the inverse scattering for \eqref{1.6} consists of 
the determination of the potentials $q$ and $r$ when the scattering coefficients and the bound-state information are known. 
When the scattering coefficients and the bound-state information contain the parameter $t$ compatible with the AKNS pair matrix 
$\mathcal T$ appearing in \eqref{1.5}, the solution to the inverse problem for \eqref{1.6} yields the potentials $q$ and $r$ that also 
contain the parameter $t.$ In that case, those time-evolved potentials $q$ and $r$ satisfy the nonlinear partial differential 
equations given in \eqref{1.1}.
 
Our paper is organized as follows. In Section~\ref{section2} we present the four Jost solutions to \eqref{1.6} and briefly
describe their relevant properties. Using the relationship \eqref{1.12}, we relate the Jost solutions to 
\eqref{1.6} to the Jost solutions to the linear system given in \eqref{1.11}. As indicated in Theorem~\ref{theorem2.1}, we 
determine that the potential pair $(q,r)$ appearing in \eqref{1.6} is related to the potential pair $(\tilde q,\tilde r)$ as in 
\eqref{2.8}. This helps us to determine the properties of the Jost solutions and the scattering coefficients for \eqref{1.6} in 
terms of the known properties \cite{AEU2023a,AEU2023b} of the corresponding quantities associated with \eqref{1.11}.
In Section~\ref{section3} we present the relevant properties of the bound-state information for \eqref{1.6} by exploiting
the connection between \eqref{1.6} and \eqref{1.11} established in Section~\ref{section2}. In Section~\ref{section4} we 
introduce a Riemann--Hilbert problem related to the inverse scattering problem for \eqref{1.6}. This is done by relating a 
pair of Jost solutions to \eqref{1.6} to another pair of Jost solutions, where the functional relationship involves the scattering 
coefficients for \eqref{1.6}. With the help of the appropriate properties of the Jost solutions and through the use of a Fourier 
transformation, we derive a Marchenko system \cite{ABDV2010,ADV2007,ADV2010,AE2019,AE2022,AV2006,B2017} of linear
integral equations associated with \eqref{1.6}.
We then show how the potentials $q$ and $r$ in \eqref{1.6} are constructed 
from the solution to the Marchenko system with the input data set consisting of the scattering coefficients and the 
bound-state information for \eqref{1.6}. In fact, we show that it is sufficient to use the right scattering coefficients in our input 
data set as the left scattering coefficients are uniquely determined by the right scattering coefficients.
In Section~\ref{section5} we obtain solutions to the nonlinear system \eqref{1.1} with the help of the solution to the 
Marchenko system with the time-evolved input data set. Finally, in Section~\ref{section6} we illustrate the solution method of 
Section~\ref{section5} with two explicit examples.

\section{The direct scattering problem}
\label{section2}

In this section we study the direct scattering problem for the linear system \eqref{1.6} associated with the Chen--Lee--Liu 
system \eqref{1.1}. We recall that for each fixed $t\in\mathbb R,$ the potentials $q$ and $r$ in
\eqref{1.6} belong to the Schwartz class $\mathcal S(\mathbb R)$ in $x\in\mathbb R.$ As already indicated, we suppress the
dependence on the parameter $t$ in the quantity related to \eqref{1.6} and \eqref{1.11}. 

We first introduce the four Jost solutions, which are the particular column-vector solutions to \eqref{1.6},
denoted by $\psi(\zeta,x),$ $\bar\psi(\zeta,x),$ $\phi(\zeta,x),$ $\bar\phi(\zeta,x),$ respectively. We use the subscripts $1$ 
and $2$ for the first and second components of the Jost solutions, respectively, i.e. we let
\begin{equation} \label{2.1}
\psi(\zeta,x)=\begin{bmatrix}
\psi_1(\zeta,x)\\
\noalign{\medskip}\psi_2(\zeta,x)
\end{bmatrix},\quad \bar\psi(\zeta,x)=\begin{bmatrix}
\bar\psi_1(\zeta,x)\\ \noalign{\medskip}\bar\psi_2(\zeta,x)
\end{bmatrix},
\end{equation}
 \begin{equation} \label{2.2}
\phi(\zeta,x)=\begin{bmatrix}
\phi_1(\zeta,x)\\
\noalign{\medskip}\phi_2(\zeta,x)
\end{bmatrix},\quad \bar\phi(\zeta,x)=\begin{bmatrix}
\bar\phi_1(\zeta,x)\\ \noalign{\medskip}\bar\phi_2(\zeta,x)
\end{bmatrix}.
\end{equation}
The Jost solutions to \eqref{1.6} are those solutions that satisfy the respective spacial asymptotics
\begin{equation}\label{2.3}
\begin{bmatrix}
\psi_1(\zeta,x)\\
\noalign{\medskip}\psi_2(\zeta,x)
\end{bmatrix}=\begin{bmatrix}
o(1)\\
\noalign{\medskip}
 e^{i\zeta^2x}\left[1+o(1)\right]
\end{bmatrix} ,\qquad  x\to+\infty,
\end{equation}
\begin{equation}\label{2.4}
\begin{bmatrix}
\bar\psi_1(\zeta,x)\\ \noalign{\medskip}\bar\psi_2(\zeta,x)
\end{bmatrix}=\begin{bmatrix}
e^{-i\zeta^2x}\left[1+o(1)\right]\\
\noalign{\medskip}
o(1)
\end{bmatrix} ,\qquad  x\to+\infty,
\end{equation}
\begin{equation}\label{2.5}
\begin{bmatrix}
\phi_1(\zeta,x)\\
\noalign{\medskip}\phi_2(\zeta,x)
\end{bmatrix}=\begin{bmatrix}
e^{-i\zeta^2x}\left[1+o(1)\right]\\
\noalign{\medskip}
o(1)
\end{bmatrix} ,\qquad   x\to-\infty,
\end{equation}
\begin{equation}\label{2.6}
\begin{bmatrix}
\bar\phi_1(\zeta,x)\\ \noalign{\medskip}\bar\phi_2(\zeta,x)
\end{bmatrix}=\begin{bmatrix}
o(1)\\
\noalign{\medskip}
e^{i\zeta^2x}\left[1+o(1)\right]
\end{bmatrix} ,\qquad  x\to-\infty.
\end{equation}
We remark that the overbar in our notation does not denote complex conjugation.

We introduce the auxiliary spectral parameter $\lambda$ in terms of the spectral parameter $\zeta$ as
\begin{equation}\label{2.7}
\lambda=\zeta^2, \quad \zeta=\sqrt{\lambda},
\end{equation} 
where the square root denotes the principal part of the complex-valued square root function. 
We use $\mathbb C$ to denote the complex plane, $\mathbb C^+$ for the upper-half complex plane, $\mathbb C^-$ for the lower-half
complex plane, and we let $\overline{\mathbb C^+}:=\mathbb C^+\cup\mathbb R$ and
$\overline{\mathbb C^-}:=\mathbb C^-\cup\mathbb R.$

In the next theorem, we describe the relevant properties of the Jost solutions to \eqref{1.6} concerning their existence and their domains of 
analyticity and continuity in $\zeta.$ Since we analyze those properties for each fixed $t\in\mathbb R,$ we do not separately mention that 
the properties hold for each fixed $t\in\mathbb R.$

\begin{theorem}
\label{theorem2.1}
Assume that the potentials $q$ and $r$ in \eqref{1.6} belong to the Schwartz class $\mathcal S(\mathbb R)$ in $x\in\mathbb R.$ Let the spectral parameter $\zeta$ is related to the parameter $\lambda$ as in \eqref{2.7}. Then, we have the following:

\begin{enumerate}

\item[\text{\rm(a)}] The Jost solutions $\psi(\zeta,x),$ and $\phi(\zeta,x)$
to \eqref{1.6} exist, are analytic in the first and third quadrants in the complex $\zeta$-plane,
and are continuous in the closures of those quadrants for each fixed $x\in\mathbb R.$ Similarly, the Jost solutions
$\bar\psi(\zeta,x)$ and $\bar\phi(\zeta,x)$ to \eqref{1.6} exist, are analytic in the 
second and fourth quadrants in the complex $\zeta$-plane, and are continuous in the closures of those quadrants for each fixed $x\in\mathbb R.$ 

\item[\text{\rm(b)}]  The four Jost solution components $\psi_1(\zeta,x),$ $\bar\psi_2(\zeta,x),$ $\bar\phi_1(\zeta,x),$ and
$\phi_2(\zeta,x)$ are odd in $\zeta$ whereas the other four Jost solution components 
$\bar\psi_1(\zeta,x),$ $\psi_2(\zeta,x),$ $\phi_1(\zeta,x),$ and  $\bar\phi_2(\zeta,x)$ are even in $\zeta.$ 
Furthermore, for each fixed $x\in\mathbb R,$ the four scalar functions $\psi_1(\zeta,x)/\zeta,$ $\psi_2(\zeta,x),$ 
$\phi_1(\zeta,x),$ $\phi_2(\zeta,x)/\zeta$ are even in 
$\zeta,$ and hence they are functions of $\lambda.$ As functions of $\lambda,$ those four scalar quantities are analytic in $\lambda \in\mathbb C^+$ and continuous in 
$\lambda\in\overline{\mathbb C^+}.$ Similarly, for each fixed $x\in\mathbb R,$ the four scalar functions 
$\bar\psi_1(\zeta,x),$ $\bar\psi_2(\zeta,x)/\zeta,$ $\bar\phi_1(\zeta,x)/\zeta,$ 
$\bar\phi_2(\zeta,x)$ are even in $\zeta,$ and hence they are functions of $\lambda.$ As functions of $\lambda,$ those four scalar quantities are analytic in 
$\lambda \in\mathbb C^-$ and continuous in $\lambda\in\overline{\mathbb C^-}.$

\end{enumerate}
\end{theorem}

\begin{proof}
The proof is obtained by proceeding as in Theorem 2.2 of \cite{AEU2023a}, where the corresponding results are obtained for 
the linear system \eqref{1.11}.
\end{proof}

In the next theorem, we present the relationship between the potential pair $(q,r)$ in \eqref{1.6} and the potential pair
$(\tilde q,\tilde r)$ in \eqref{1.11}, when the solutions to \eqref{1.6} and to \eqref{1.11} are related to each other as
in \eqref{1.12}.

\begin{theorem}
\label{theorem2.2}

Assume that the potentials $q$ and $r$ appearing in \eqref{1.6} belong to the Schwartz class $\mathcal S(\mathbb R)$
in $x\in\mathbb R.$ Further assume that solutions to \eqref{1.11} are related to the solutions to \eqref{1.6} as in \eqref{1.12}.
Then, the potential pair $(q,r)$ is related to the potential pair $(\tilde q, \tilde r)$ as
\begin{equation}\label{2.8}
q(x)=E(x)^{-1}\tilde q(x),\quad  r(x)=E(x)\,\tilde r(x),
\end{equation}
where $E(x)$ is the quantity defined in \eqref{1.13}. Consequently, the potential pair $(\tilde q,\tilde r)$ also belongs to
the Schwartz class $\mathcal S(\mathbb R).$ 
\end{theorem}

\begin{proof}Taking the $x$-derivative on both sides in \eqref{1.12}, we obtain
\begin{equation}\label{2.9}
\begin{bmatrix}\alpha\\
\noalign{\medskip}
\beta\end{bmatrix}'
=c\begin{bmatrix}
0&0\\
\noalign{\medskip}
0&E'(x)\end{bmatrix}
\begin{bmatrix}\tilde\alpha\\ 
\noalign{\medskip}
\tilde\beta\end{bmatrix}
+c\begin{bmatrix}
1&0\\
\noalign{\medskip}
0&E(x)
\end{bmatrix}
\begin{bmatrix}\tilde\alpha\\ 
\noalign{\medskip}
\tilde\beta\end{bmatrix}',
\end{equation}
where we recall that $c$ is a constant independent of $x$ but may depend on the spectral parameter $\zeta.$ 
Using \eqref{1.6} on the left-hand side of \eqref{2.9} and using \eqref{1.11} on the right-hand side, we obtain
\begin{equation}
\label{2.10}
\begin{bmatrix}
-i\zeta^2 & \zeta \,q(x)\\
\noalign{\medskip}
\zeta\, r(x) & i\zeta^2+\displaystyle\frac{i}{2}\,q(x) r(x)
\end{bmatrix}
\begin{bmatrix}
\alpha\\
\noalign{\medskip}
\beta
\end{bmatrix}=c\begin{bmatrix}
0&0\\
\noalign{\medskip}
0&E'(x)\end{bmatrix}
\begin{bmatrix}\tilde\alpha\\ 
\noalign{\medskip}
\tilde\beta\end{bmatrix}
+c\begin{bmatrix}
1&0\\
\noalign{\medskip}
0&E(x)
\end{bmatrix}
\begin{bmatrix}
-i\zeta^2& \zeta\, \tilde q(x)
\\
\noalign{\medskip}
\zeta\, \tilde r(x)&i\zeta^2
\end{bmatrix}
\begin{bmatrix}
\tilde\alpha\\
\noalign{\medskip}
\tilde\beta
\end{bmatrix}.
\end{equation}
Next, using \eqref{1.12} on the left-hand side of \eqref{2.10}, we get
\begin{equation}\label{2.11}
\begin{split}
c\begin{bmatrix}
-i\zeta^2 & \zeta \,q(x)\\
\noalign{\medskip}
\zeta\, r(x) & i\zeta^2+\displaystyle\frac{i}{2}\,q(x) r(x)
\end{bmatrix}
\begin{bmatrix}
1& 0\\
\noalign{\medskip}
0&E(x)
\end{bmatrix}\begin{bmatrix}
\tilde\alpha\\
\noalign{\medskip}
\tilde\beta
\end{bmatrix}=&\,c\begin{bmatrix}
0&0\\
\noalign{\medskip}
0&E'(x)\end{bmatrix}
\begin{bmatrix}\tilde\alpha\\ 
\noalign{\medskip}
\tilde\beta\end{bmatrix}\\
&+c
\begin{bmatrix}
-i\zeta^2& \zeta\, \tilde q(x)
\\
\noalign{\medskip}
\zeta \tilde r(x)E(x)&i\zeta^2E(x)
\end{bmatrix}
\begin{bmatrix}
\tilde\alpha\\
\noalign{\medskip}
\tilde\beta
\end{bmatrix}.
\end{split}
\end{equation}
Since $c\begin{bmatrix}
\tilde\alpha\\
\noalign{\medskip}
\tilde\beta
\end{bmatrix}$ is an arbitrary solution to \eqref{1.11}, from \eqref{2.11} we get 
\begin{equation}\label{2.12}
\begin{bmatrix}
-i\zeta^2 & \zeta q(x)E(x)\\
\noalign{\medskip}
\zeta\, r(x) & i\zeta^2 E(x)+\displaystyle\frac{i}{2}\, q(x) r(x)E(x)
\end{bmatrix}
=
\begin{bmatrix}
-i\zeta^2& \zeta\, \tilde q(x)
\\
\noalign{\medskip}
\zeta\tilde r(x)E(x)&E'(x)+i\zeta^2 E(x)
\end{bmatrix}.
\end{equation}
By taking the $x$-derivative of both sides of \eqref{1.13} we have
\begin{equation}\label{2.13}
E'(x)=\displaystyle\frac{i}{2}\,q(x) r(x)E(x).
\end{equation}
Comparing the individual entries of the matrix equality in \eqref{2.12}, with the help of \eqref{2.13}, we observe that \eqref{2.8}
holds. With the help of \eqref{1.13} we confirm that the potentials $\tilde q$ and $\tilde r$ belong to the Schwartz class
$\mathcal S(\mathbb R)$ when the potentials $q$ and $r$ belong to $\mathcal S(\mathbb R).$ 
\end{proof}

As indicated in Theorem~\ref{theorem2.2}, the potential pair $(\tilde q,\tilde r)$ appearing in \eqref{1.11} belongs to the Schwartz
class $\mathcal S(\mathbb R)$ when the potential pair $(q,r)$ appearing in \eqref{1.6} belongs to $\mathcal S(\mathbb R)$.
We introduce the Jost solutions to \eqref{1.11} in the same manner the Jost solutions to \eqref{1.6} are introduced. For the Jost 
solutions associated with \eqref{1.11}, we use the notation $\tilde \psi(\zeta,x),$ $\tilde{\bar\psi}(\zeta,x),$ $\tilde\phi(\zeta,x),$ and
$\tilde{\bar\phi}(\zeta,x)$ by requiring that those solutions satisfy the respective spacial asymptotics
\begin{equation}
\label{2.14}
\begin{bmatrix}
\tilde\psi_1(\zeta,x)\\
\noalign{\medskip}\tilde\psi_2(\zeta,x)
\end{bmatrix}=\begin{bmatrix}
o(1)\\
\noalign{\medskip}
 e^{i\zeta^2x}\left[1+o(1)\right]
\end{bmatrix} ,\qquad  x\to+\infty,
\end{equation}
\begin{equation}
\label{2.15}
\begin{bmatrix}
\tilde{\bar\psi}_1(\zeta,x)\\ \noalign{\medskip}\tilde{\bar\psi}_2(\zeta,x)
\end{bmatrix}=\begin{bmatrix}
e^{-i\zeta^2x}\left[1+o(1)\right]\\
\noalign{\medskip}
o(1)
\end{bmatrix} ,\qquad  x\to+\infty,
\end{equation}
\begin{equation}
\label{2.16}
\begin{bmatrix}
\tilde\phi_1(\zeta,x)\\
\noalign{\medskip}\tilde\phi_2(\zeta,x)
\end{bmatrix}=\begin{bmatrix}
e^{-i\zeta^2x}\left[1+o(1)\right]\\
\noalign{\medskip}
o(1)
\end{bmatrix} ,\qquad   x\to-\infty,
\end{equation}
\begin{equation}
\label{2.17}
\begin{bmatrix}
\tilde{\bar\phi}_1(\zeta,x)\\ \noalign{\medskip}\tilde{\bar\phi}_2(\zeta,x)
\end{bmatrix}=\begin{bmatrix}
o(1)\\
\noalign{\medskip}
e^{i\zeta^2x}\left[1+o(1)\right]
\end{bmatrix} ,\qquad  x\to-\infty,
\end{equation}
which are the analogs of \eqref{2.3}--\eqref{2.6}, respectively.
We refer the reader to \cite{AEU2023a, AEU2023b} where it is shown that the Jost solutions to \eqref{1.11} have the same
properties listed in Theorem~\ref{theorem2.1} and satisfied by the Jost solutions to \eqref{1.6}.

In the following theorem, we show the connection between the Jost solutions to \eqref{1.6} and
the Jost solutions to \eqref{1.11} 

\begin{theorem}
\label{theorem2.3}
Let the potential pair $(q,r)$ in \eqref{1.6} and the potential pair $(\tilde q,\tilde r)$ in \eqref{1.11} are related to each other as in
\eqref{2.8}, where $E(x)$ and $\mu$ are the quantities appearing in \eqref{1.13} and \eqref{1.15}, respectively.
Assume that the potential pair $(q,r)$ belongs to the Schwartz class $\mathcal S(\mathbb R).$ Then, we have the following:

\begin{enumerate}

\item[\text{\rm(a)}]
The Jost solution $\psi(\zeta,x)$ 
to \eqref{1.6} is related to the Jost solution $\tilde \psi(\zeta,x)$ to \eqref{1.11} as
\begin{equation}\label{2.18}
\begin{bmatrix}
\psi_1(\zeta,x)\\
\noalign{\medskip}\psi_2(\zeta,x)
\end{bmatrix}=e^{-i \mu/2}\begin{bmatrix}
1& 0\\
\noalign{\medskip}
0&E(x)
\end{bmatrix}\begin{bmatrix}
\tilde\psi_1(\zeta,x)\\
\noalign{\medskip}\tilde\psi_2(\zeta,x)
\end{bmatrix},
\end{equation}
where $\tilde\psi_1(\zeta,x)$ and $\tilde\psi_2(\zeta,x)$ denote the first and second components of the Jost solution
$\tilde\psi(\zeta,x)$ in a manner similar to the first equality of \eqref{2.1}.

\item[\text{\rm(b)}]
The Jost solution $\bar\psi(\zeta,x)$ 
to \eqref{1.6} is related to the Jost solution $\tilde{\bar\psi}(\zeta,x)$ to \eqref{1.11} as
\begin{equation}\label{2.19}
\begin{bmatrix}
\bar\psi_1(\zeta,x)\\ \noalign{\medskip}\bar\psi_2(\zeta,x)
\end{bmatrix}=\begin{bmatrix}
1& 0\\
\noalign{\medskip}
0&E(x)
\end{bmatrix}\begin{bmatrix}
\tilde{\bar\psi}_1(\zeta,x)\\ \noalign{\medskip}\tilde{\bar\psi}_2(\zeta,x)
\end{bmatrix},
\end{equation}
where $\tilde{\bar\psi}_1(\zeta,x)$ and $\tilde{\bar\psi}_2(\zeta,x)$ denote the first and second components of the Jost solution
$\tilde{\bar\psi}(\zeta,x)$ in a manner similar to the second equality of \eqref{2.1}.

\item[\text{\rm(c)}]
The Jost solution $\phi(\zeta,x)$ 
to \eqref{1.6} is related to the Jost solution $\tilde \phi(\zeta,x)$ to \eqref{1.11} as
\begin{equation}\label{2.20}
\begin{bmatrix}
\phi_1(\zeta,x)\\
\noalign{\medskip}\phi_2(\zeta,x)
\end{bmatrix}=\begin{bmatrix}
1& 0\\
\noalign{\medskip}
0&E(x)
\end{bmatrix}\begin{bmatrix}
\tilde\phi_1(\zeta,x)\\
\noalign{\medskip}\tilde\phi_2(\zeta,x)
\end{bmatrix},
\end{equation}
where $\tilde\phi_1(\zeta,x)$ and $\tilde\phi_2(\zeta,x)$ denote the first and second components of the Jost solution
$\tilde\phi(\zeta,x)$ in a manner similar to the first equality of \eqref{2.2}.

\item[\text{\rm(d)}]
The Jost solution $\bar\phi(\zeta,x)$ 
to \eqref{1.6} is related to the Jost solution $\tilde{\bar\phi}(\zeta,x)$ to \eqref{1.11} as
\begin{equation}\label{2.21}
\begin{bmatrix}
\bar\phi_1(\zeta,x)\\ \noalign{\medskip}\bar\phi_2(\zeta,x)
\end{bmatrix}=\begin{bmatrix}
1& 0\\
\noalign{\medskip}
0&E(x)
\end{bmatrix}\begin{bmatrix}
\tilde{\bar\phi}_1(\zeta,x)\\ \noalign{\medskip}\tilde{\bar\phi}_2(\zeta,x)
\end{bmatrix},
\end{equation}
where $\tilde{\bar\phi}_1(\zeta,x)$ and $\tilde{\bar\phi}_2(\zeta,x)$ denote the first and second components of the Jost solution
$\tilde{\bar\phi}(\zeta,x)$ in a manner similar to the second equality of \eqref{2.2}.
\end{enumerate}
\end{theorem}

\begin{proof} We apply the relationship to \eqref{1.12} to the corresponding Jost solutions to \eqref{1.6} and \eqref{1.11}, respectively,
and in each case we determine the specific value of the constant $c$ appearing in \eqref{1.12} for each pair of the Jost solutions.
For example, to establish \eqref{2.18} we proceed as follows. We write \eqref{1.12} by using the corresponding Jost solutions 
$\psi(\zeta,x)$ and $\tilde\psi(\zeta,x)$ there. This yields
\begin{equation}\label{2.22}
\begin{bmatrix}
\psi_1(\zeta,x)\\
\noalign{\medskip}\psi_2(\zeta,x)
\end{bmatrix}
=c\begin{bmatrix}
1&0\\
\noalign{\medskip}
0&E(x)
\end{bmatrix}
\begin{bmatrix}
\tilde\psi_1(\zeta,x)\\
\noalign{\medskip}
\tilde\psi_2(\zeta,x)
\end{bmatrix}.
\end{equation}
By letting $x\to +\infty$ in \eqref{2.22}, with the help of the asymptotics in \eqref{2.3} and \eqref{2.14} and the second asymptotics in
\eqref{1.14}, we obtain $c$ as $e^{-i\mu/2}.$ The relationships presented in (b), (c), and (d) are obtained in a similar manner with the 
help of the spacial asymptotics in \eqref{1.14}, \eqref{2.4}--\eqref{2.6}, and \eqref{2.15}--\eqref{2.17}. 
\end{proof}

In the next theorem, we present the large $\zeta$-asymptotics of the Jost solutions to \eqref{1.6}.
In the theorem, those asymptotics are expressed in terms of $\lambda,$ which is related to $\zeta$ as in \eqref{2.7}. 

\begin{theorem}
\label{theorem2.4}
Assume that the potentials $q$ and $r$ in \eqref{1.6} belong to the Schwartz 
class $\mathcal S(\mathbb R).$ Let the parameter
$\lambda$ be related to the spectral parameter $\zeta$ as in \eqref{2.7}.
Then, for each fixed $x\in\mathbb R,$ as $\lambda\to\infty$ in $\overline{\mathbb C^+}$
the Jost solutions $\psi(\zeta,x)$ and $\phi(\zeta,x)$ to \eqref{1.6}
appearing in \eqref{2.3} and \eqref{2.5}, respectively, satisfy
\begin{equation}\label{2.23}
\displaystyle\frac{\psi_1(\zeta,x)}{\zeta}=
e^{i\lambda x}\left[\displaystyle\frac{q(x)}{2i\lambda} +O\left(\displaystyle\frac{1}{\lambda^2}\right)\right],
\end{equation}
\begin{equation}
\label{2.24}
\psi_2(\zeta,x)=e^{i\lambda x}\left[1+\displaystyle\frac{q(x)\,r(x)}{4\lambda}
+\displaystyle\frac{1}{4\lambda}\int_x^\infty dy\,q(y)\,r'(y)+O\left(\frac{1}{\lambda^2}\right)\right],
\end{equation}
\begin{equation}
\label{2.25}
\phi_1(\zeta,x)=
\displaystyle e^{-i\lambda x} E(x)\left[1+\displaystyle\frac{1}{4\lambda}\int_{-\infty}^x dy\,q(y)\,r'(y)
+O\left(\displaystyle\frac{1}{\lambda^{2}}\right)\right],
\end{equation}
\begin{equation}
\label{2.26}
\displaystyle\frac{\phi_2(\zeta,x)}{\zeta}=e^{-i\lambda x}E(x)
\left[\displaystyle\frac{i\,r(x)}{2\lambda}+O\left(\displaystyle\frac{1}{\lambda^2}\right)\right],
\end{equation}
where we recall that $E(x)$ and $\mu$ are the quantities in \eqref{1.13} and \eqref{1.15}, respectively. Similarly, for each fixed $x\in\mathbb R,$ 
as $\lambda\to\infty$ in $\overline{\mathbb C^-}$ the Jost solutions $\bar\psi(\zeta,x)$ and $\bar\phi(\zeta,x)$
to \eqref{1.6}
appearing in \eqref{2.4} and \eqref{2.6}, respectively, satisfy
\begin{equation}
\label{2.27}
\bar\psi_1(\zeta,x)=
\displaystyle e^{-i\mu/2-i\lambda x} E(x)
\left[1-\displaystyle\frac{1}{4\lambda}\int_x^\infty dy\,q(y)\,r'(y)
+O\left(\displaystyle\frac{1}{\lambda^{2}}\right)\right],
\end{equation}
\begin{equation}\label{2.28}
\displaystyle\frac{\bar\psi_2(\zeta,x)}{\zeta}=
e^{-i\mu/2-i\lambda x}E(x)\left[\displaystyle\frac{i\,r(x)}{2\lambda}+O\left(\displaystyle\frac{1}{\lambda^2}\right)\right],
\end{equation}
\begin{equation}
\label{2.29}
\displaystyle\frac{\bar\phi_1(\zeta,x)}{\zeta}=
e^{i\lambda x}\left[\displaystyle\frac{q(x)}{2i\lambda}+O\left(\displaystyle\frac{1}{\lambda^2}\right)\right],
\end{equation}
\begin{equation}
\label{2.30}
\bar\phi_2(\zeta,x)=
e^{i\lambda x}\left[1+\displaystyle\frac{q(x)\,r(x)}{4\lambda}
-\displaystyle\frac{1}{4\lambda}\int_{-\infty}^{x}dy\,q(y)\,r'(y)
+O\left(\frac{1}{\lambda^2}\right)\right].
\end{equation}
	
\end{theorem}

\begin{proof} The large $\zeta$-asymptotics of the Jost solutions to \eqref{1.11} are already known \cite{AEU2023a,AEU2023b} and they 
are listed in Theorem 2.4 in \cite{AEU2023a}. We use
those asymptotics on the right-hand sides of \eqref{2.18}--\eqref{2.21}. We also express $\tilde q$ and $\tilde r$ appearing in those
asymptotics in terms of $q$ and $r$ with the help of \eqref{2.8}. Furthermore, we use \eqref{2.7} to relate the parameters $\lambda$
and $\zeta$ to each other. We then obtain the large $\zeta$-asymptotics expressed in \eqref{2.23}--\eqref{2.30}.
\end{proof}

Next, we introduce the scattering coefficients associated with the linear system \eqref{1.6} by using the spacial asymptotics of
the Jost solutions to \eqref{1.6} as
\begin{equation}\label{2.31}
\begin{bmatrix}
\psi_1(\zeta,x)\\
\noalign{\medskip}\psi_2(\zeta,x)
\end{bmatrix}=\begin{bmatrix}
\displaystyle\frac{L(\zeta)}{T_{\text{\rm{l}}}(\zeta)}\,e^{-i\zeta^2 x}\left[1+o(1)\right]\\
\noalign{\medskip}
\displaystyle\frac{1}{T_{\text{\rm{l}}}(\zeta)}\,e^{i\zeta^2 x}\left[1+o(1)\right]
\end{bmatrix}, \qquad   x\to-\infty,
\end{equation}
\begin{equation}\label{2.32}
\begin{bmatrix}
\bar\psi_1(\zeta,x)\\ \noalign{\medskip}\bar\psi_2(\zeta,x)
\end{bmatrix}=\begin{bmatrix}
\displaystyle\frac{1}{\bar{T_{\text{\rm{l}}}}(\zeta)}\,e^{-i\zeta^2 x}\left[1+o(1)\right]\\
\noalign{\medskip}
\displaystyle\frac{\bar L(\zeta)}{\bar{T_{\text{\rm{l}}}}(\zeta)}\,e^{i\zeta^2 x}\left[1+o(1)\right]
\end{bmatrix}, \qquad  x\to-\infty,
\end{equation}
\begin{equation}\label{2.33}
\begin{bmatrix}
\phi_1(\zeta,x)\\
\noalign{\medskip}\phi_2(\zeta,x)
\end{bmatrix}=\begin{bmatrix}
\displaystyle\frac{1}{T_{\text{\rm{r}}}(\zeta)}\,e^{-i\zeta^2x}\left[1+o(1)\right]\\
\noalign{\medskip}
\displaystyle\frac{R(\zeta)}{T_{\text{\rm{r}}}(\zeta)}\,e^{i\zeta^2 x}\left[1+o(1)\right]
\end{bmatrix}, \qquad   x\to+\infty,
\end{equation}
\begin{equation}\label{2.34}
\begin{bmatrix}
\bar\phi_1(\zeta,x)\\ \noalign{\medskip}\bar\phi_2(\zeta,x)
\end{bmatrix}=\begin{bmatrix}
\displaystyle\frac{\bar R(\zeta)}   {\bar{T_{\text{\rm{r}}}}(\zeta)}\,e^{-i\zeta^2 x}\left[1+o(1)\right]\\
\noalign{\medskip}
\displaystyle\frac{1}{\bar{T_{\text{\rm{r}}}}(\zeta)}\,e^{i\zeta^2 x}\left[1+o(1)\right]
\end{bmatrix}, \qquad   x\to+\infty.
\end{equation}
We refer to $T_{\text{\rm{l}}}(\zeta)$ and $\bar{T_{\text{\rm{l}}}}(\zeta)$ as the left transmission coefficients,
$L(\zeta)$ and $\bar L(\zeta)$ as the left reflection coefficients,
$T_{\text{\rm{r}}}(\zeta)$ and $\bar{T_{\text{\rm{r}}}}(\zeta)$ as the right transmission coefficients, and
$R(\zeta)$ and $\bar R(\zeta)$ as the right reflection coefficients.

Similarly, we introduce the scattering coefficients associated with the linear system \eqref{1.11} by using the spacial asymptotics of the Jost solutions to \eqref{1.11} as
\begin{equation}\label{2.35}
\begin{bmatrix}
\tilde\psi_1(\zeta,x)\\
\noalign{\medskip}\tilde\psi_2(\zeta,x)
\end{bmatrix}=\begin{bmatrix}
\displaystyle\frac{\tilde L(\zeta)}{\tilde T(\zeta)}\,e^{-i\zeta^2 x}\left[1+o(1)\right]\\
\noalign{\medskip}
\displaystyle\frac{1}{\tilde T(\zeta)}\,e^{i\zeta^2 x}\left[1+o(1)\right]
\end{bmatrix}, \qquad   x\to-\infty,
\end{equation}
\begin{equation}\label{2.36}
\begin{bmatrix}
\tilde{\bar\psi}_1(\zeta,x)\\ \noalign{\medskip}\tilde{\bar\psi}_2(\zeta,x)
\end{bmatrix}=\begin{bmatrix}
\displaystyle\frac{1}{\tilde{\bar T}(\zeta)}\,e^{-i\zeta^2 x}\left[1+o(1)\right]\\
\noalign{\medskip}
\displaystyle\frac{\tilde{\bar L}(\zeta)}{\tilde{\bar T}(\zeta)}\,e^{i\zeta^2 x}\left[1+o(1)\right]
\end{bmatrix}, \qquad  x\to-\infty,
\end{equation}
\begin{equation}\label{2.37}
\begin{bmatrix}
\tilde\phi_1(\zeta,x)\\
\noalign{\medskip}\tilde\phi_2(\zeta,x)
\end{bmatrix}=\begin{bmatrix}
\displaystyle\frac{1}{\tilde T(\zeta)}\,e^{-i\zeta^2x}\left[1+o(1)\right]\\
\noalign{\medskip}
\displaystyle\frac{\tilde R(\zeta)}{\tilde T(\zeta)}\,e^{i\zeta^2 x}\left[1+o(1)\right]
\end{bmatrix}, \qquad   x\to+\infty,
\end{equation}
\begin{equation}
\label{2.38}
\begin{bmatrix}
\tilde{\bar\phi}_1(\zeta,x)\\ \noalign{\medskip}\tilde{\bar\phi}_2(\zeta,x)
\end{bmatrix}=\begin{bmatrix}
\displaystyle\frac{\tilde{\bar R}(\zeta)}   {\tilde{\bar T}(\zeta)}\,e^{-i\zeta^2 x}\left[1+o(1)\right]\\
\noalign{\medskip}
\displaystyle\frac{1}{\tilde{\bar T}(\zeta)}\,e^{i\zeta^2 x}\left[1+o(1)\right]
\end{bmatrix}, \qquad   x\to+\infty.
\end{equation}
We refer to $\tilde T(\zeta)$ and $\tilde{\bar T}(\zeta)$ as the transmission coefficients,
$L(\zeta)$ and $\bar L(\zeta)$ as the left reflection coefficients, and
$R(\zeta)$ and $\bar R(\zeta)$ as the right reflection coefficients.
We remark that we do not need to make a distinction between the left and right transmission coefficients for \eqref{1.11}. This is
due to the fact that the trace of the matrix $\tilde{\mathcal X}$ appearing in \eqref{1.11} is zero. Consequently, the left and right transmission 
coefficients for \eqref{1.11} coincide, and we use $T(\zeta)$ and $\bar T(\zeta)$ to denote those common values, respectively. On the
other hand, the trace of the matrix $\mathcal X$ appearing in \eqref{1.6} is not zero. This results in the fact that the left and right transmission
coefficients are unequal. Hence, to describe the transmission coefficients for \eqref{1.6}, we need to use the four quantities 
$T_{\text{\rm{l}}}(\zeta),$ $T_{\text{\rm{r}}}(\zeta),$ $\bar{T_{\text{\rm{l}}}}(\zeta),$ and $\bar{T_{\text{\rm{r}}}}(\zeta).$

In the next theorem, we show the connection among the scattering coefficients for \eqref{1.6} and
the scattering coefficients for \eqref{1.11}.

\begin{theorem}
\label{theorem2.5}
Let the potential pair $(q,r)$ in \eqref{1.6} and the potential pair $(\tilde q,\tilde r)$ in \eqref{1.11} are related to each other as in
\eqref{2.8}, where $\mu$ is the quantity appearing in \eqref{1.15}.
Assume that the potential pair $(q,r)$ belongs to the Schwartz class $\mathcal S(\mathbb R).$ Then, we have the following:
\begin{enumerate}

\item[\text{\rm(a)}] The eight scattering coefficients $T_{\text{\rm{l}}}(\zeta),$ $T_{\text{\rm{r}}}(\zeta),$
$\bar{T_{\text{\rm{l}}}}(\zeta),$ $\bar{T_{\text{\rm{r}}}}(\zeta),$ $R(\zeta),$ $L(\zeta),$ $\bar R(\zeta),$ $\bar L(\zeta)$ for 
\eqref{1.6} are related to the six scattering coefficients $\tilde T(\zeta),$ $\tilde{\bar T}(\zeta),$ $\tilde R(\zeta),$ $\tilde L(\zeta),$ 
$\tilde{\bar R}(\zeta),$ $\tilde{\bar L}(\zeta)$ for \eqref{1.11} as
\begin{equation}\label{2.39}
 T_{\text{\rm l}}(\zeta)=e^{i \mu/2} \,\tilde T(\zeta),\quad
\bar T_{\text{\rm l}}(\zeta)=
\tilde{\bar T}(\zeta),
\end{equation}
\begin{equation}\label{2.40}
T_{\text{\rm r}}(\zeta)= \,\tilde T(\zeta),\quad
\bar T_{\text{\rm r}}(\zeta)=
e^{-i\mu/2} \,\tilde{\bar T}(\zeta),
\end{equation}
\begin{equation}\label{2.41}
R(\zeta)=e^{i\mu/2} \,\tilde R(\zeta),\quad
\bar R(\zeta)=
e^{-i  \mu/2} \,\tilde{\bar R}(\zeta),
\end{equation}
\begin{equation}\label{2.42}
L(\zeta)=\tilde L(\zeta),\quad
\bar L(\zeta)=
\tilde{\bar L}(\zeta).
\end{equation}

\item[\text{\rm(b)}] The transmission coefficients $T_{\text{\rm{l}}}(\zeta)$ and $T_{\text{\rm{r}}}(\zeta)$ are even in $\zeta$, and
hence they are functions of $\lambda.$ As functions of $\lambda,$ the quantities $T_{\text{\rm{l}}}(\zeta)$ and
$T_{\text{\rm{r}}}(\zeta)$ are meromorphic in $\lambda\in\mathbb C^+$ and continuous in $\lambda\in\overline{\mathbb C^+}$ 
except at the poles causing the meromorphic property in $\mathbb C^+.$ Similarly, the transmission coefficients
$\bar{T_{\text{\rm{l}}}}(\zeta)$ and $\bar{T_{\text{\rm{r}}}}(\zeta)$ are even in $\zeta$, and
hence they are functions of $\lambda.$ As functions of $\lambda,$ the quantities $\bar{T_{\text{\rm{l}}}}(\zeta)$ and
$\bar{T_{\text{\rm{r}}}}(\zeta)$ are meromorphic in $\lambda\in\mathbb C^-$ and continuous in $\lambda\in\overline{\mathbb C^-}$ 
except at the poles causing the meromorphic property in $\mathbb C^-.$ The four quantities
$R(\zeta)/\zeta,$ $\bar R(\zeta)/\zeta,$ $L(\zeta)/\zeta,$ $\bar L(\zeta)/\zeta$ are even in $\zeta,$ and hence they are all 
functions of $\lambda.$ As functions of $\lambda,$ those four quantities are continuous in $\lambda\in\mathbb R.$ 
\end{enumerate}
\end{theorem}

\begin{proof}
For the proof of (a), we use the relationships \eqref{2.18}--\eqref{2.21} connecting to Jost solutions to \eqref{1.6} and the Jost 
solutions to \eqref{1.11}. With the help of the asymptotics in \eqref{1.14} and \eqref{2.31}--\eqref{2.38}, by comparing the leading 
terms in the spacial asymptotics of \eqref{2.18}--\eqref{2.21}, we obtain \eqref{2.39}--\eqref{2.42}. Hence, the proof of (a) is complete. 
By using the relevant properties of the scattering coefficients $T(\zeta),$ $\bar T(\zeta),$ $L(\zeta),$ $\bar L(\zeta),$ 
$R(\zeta),$ and $\bar R(\zeta)$ given in Theorem 2.2 of \cite{AEU2023a}, we establish the aforementioned properties of the scattering coefficients for
\eqref{1.6}. 
\end{proof}

In the next theorem, we present the small $\zeta$-asymptotics of the scattering coefficients for \eqref{1.6}.

\begin{theorem}
\label{theorem2.6}
Assume that the potentials $q$ and $r$ in \eqref{1.6} belong to the Schwartz class $\mathcal S(\mathbb R)$ in $x\in\mathbb R.$ Let the parameter $\lambda$ be related 
to the spectral parameter $\zeta$ as in \eqref{2.7}. Then, the small $\zeta$-asymptotics of the scattering coefficients 
$T_{\text{\rm{l}}}(\zeta),$ $T_{\text{\rm{r}}}(\zeta),$ $\bar{T_{\text{\rm{l}}}}(\zeta),$ $\bar{T_{\text{\rm{r}}}}(\zeta),$ $R(\zeta),$ $\bar R(\zeta),$ $L(\zeta),$ and $\bar L(\zeta)$ appearing in \eqref{2.31}--\eqref{2.34} are expressed in $\lambda$ as
\begin{equation}
\label{2.43}
T_{\text{\rm{l}}}(\zeta)=e^{i\mu/2}\left[1+O(\lambda)\right],\qquad \lambda\to 0
\text{\rm{ in }} \overline{\mathbb C^+},
\end{equation}
\begin{equation}
\label{2.44}
T_{\text{\rm{r}}}(\zeta)=1+O(\lambda),\qquad \lambda\to 0
\text{\rm{ in }} \overline{\mathbb C^+},
\end{equation}
\begin{equation}
\label{2.45}
\bar{T_{\text{\rm{l}}}}(\zeta)=1+O(\lambda),\qquad \lambda\to 0
\text{\rm{ in }} \overline{\mathbb C^+},
\end{equation}
\begin{equation}
\label{2.46}
\bar{T_{\text{\rm{r}}}}(\zeta)=e^{-i\mu/2}\left[1+O(\lambda)\right],\qquad \lambda\to 0
\text{\rm{ in }} \overline{\mathbb C^+},
\end{equation}
\begin{equation}
\label{2.47}
\displaystyle\frac{R(\zeta)}{\zeta}=e^{i\mu/2}\left[\displaystyle\frac{1}{E(x)}\displaystyle\int_{-\infty}^\infty dy\,r(y)+O(\lambda)\right], \qquad \lambda\to 0
\text{\rm{ in }} \mathbb R,
\end{equation}
\begin{equation}
\label{2.48}
\displaystyle\frac{\bar R(\zeta)}{\zeta}=e^{-i\mu/2}\left[E(x)\displaystyle\int_{-\infty}^\infty dy\,q(y)+O(\lambda)\right], \qquad \lambda\to 0
\text{\rm{ in }} \mathbb R,
\end{equation}
\begin{equation}
\label{2.49}
\displaystyle\frac{L(\zeta)}{\zeta}=-\left[E(x)\displaystyle\int_{-\infty}^\infty dy\,q(y)+O(\lambda)\right],\qquad \lambda\to 0
\text{\rm{ in }} \mathbb R,
\end{equation}
\begin{equation}
\label{2.50}
\displaystyle\frac{\bar L(\zeta)}{\zeta}=-\left[\displaystyle\frac{1}{E(x)}\displaystyle\int_{-\infty}^\infty dy\,r(y)+O(\lambda)\right], \qquad \lambda\to 0
\text{\rm{ in }} \mathbb R.
\end{equation}
\end{theorem}

\begin{proof} In order to obtain \eqref{2.43}--\eqref{2.50}, we use the connections among the scattering coefficients for \eqref{1.6}
and for \eqref{1.11} given in \eqref{2.39}--\eqref{2.42}. We utilize the help of \eqref{2.8} and the known small $\zeta$ asymptotics of 
the scattering coefficients for \eqref{1.11}, where those asymptotics are listed in Theorem~2.5 (d) of \cite{AEU2023a}. We then get 
\eqref{2.43}--\eqref{2.50}.
\end{proof}

In the next theorem, we present the large $\zeta$-asymptotics of the scattering coefficients for \eqref{1.6}.
In the theorem, those asymptotics are expressed in terms of $\lambda,$ which is related to $\zeta$ as in \eqref{2.7}. 

\begin{theorem}
\label{theorem2.7}
Assume that the potentials $q$ and $r$ in \eqref{1.6} belong to the Schwartz class $\mathcal S(\mathbb R)$ in $x\in\mathbb R.$ Let the parameter $\lambda$ be related 
to the spectral parameter $\zeta$ as in \eqref{2.7}. Then, the large $\zeta$-asymptotics of the scattering coefficients 
$T_{\text{\rm{l}}}(\zeta),$ $T_{\text{\rm{r}}}(\zeta),$ $\bar{T_{\text{\rm{l}}}}(\zeta),$ $\bar{T_{\text{\rm{r}}}}(\zeta),$ $R(\zeta),$ $\bar R(\zeta),$ $L(\zeta),$ and $\bar L(\zeta)$ appearing in \eqref{2.31}--\eqref{2.34} are expressed in $\lambda$ as
\begin{equation}
\label{2.51}
T_{\text{\rm{l}}}(\zeta)=1+O\left(\displaystyle\frac{1}{\lambda}\right),\qquad \lambda\to \infty
\text{\rm{ in }} \overline{\mathbb C^+},
\end{equation}
\begin{equation}
\label{2.52}
\bar{T_{\text{\rm{l}}}}(\zeta)=e^{i\mu/2}\left[1+O\left(\displaystyle\frac{1}{\lambda}\right)\right],\qquad \lambda\to \infty
\text{\rm{ in }} \overline{\mathbb C^+},
\end{equation}
\begin{equation}
\label{2.53}
T_{\text{\rm{r}}}(\zeta)=e^{-i\mu/2}\left[1+O\left(\displaystyle\frac{1}{\lambda}\right)\right],\qquad \lambda\to \infty
\text{\rm{ in }} \overline{\mathbb C^+},
\end{equation}
\begin{equation}
\label{2.54}
\bar{T_{\text{\rm{r}}}}(\zeta)=1+O\left(\displaystyle\frac{1}{\lambda}\right),\qquad \lambda\to \infty
\text{\rm{ in }} \overline{\mathbb C^+},
\end{equation}
\begin{equation}
\label{2.55}
R(\zeta)=O\left(\displaystyle\frac{1}{\zeta^3}\right), \quad \bar R(\zeta)=O\left(\displaystyle\frac{1}{\zeta^3}\right), \qquad \lambda\to\pm \infty,
\end{equation}
\begin{equation}
\label{2.56}
L(\zeta)=O\left(\displaystyle\frac{1}{\zeta^3}\right),\quad \bar L(\zeta)=O\left(\displaystyle\frac{1}{\zeta^3}\right), \qquad \lambda\to\pm \infty.
\end{equation}

\end{theorem}

\begin{proof}
The large $\zeta$-asymptotics of the scattering coefficients for \eqref{1.11} are known and they are listed in
(2.46)--(2.51) of \cite{AEU2023a}. Using those asymptotics in \eqref{2.39}--\eqref{2.42}, we obtain \eqref{2.51}--\eqref{2.56}.
\end{proof}

In the next theorem we show that the left scattering coefficients $T_{\text{\rm{l}}}(\zeta),$ $\bar{T_{\text{\rm{l}}}}(\zeta),$ $L(\zeta),$
$\bar L(\zeta)$ can be expressed in terms of the right scattering coefficients $T_{\text{\rm{r}}}(\zeta),$ $\bar{T_{\text{\rm{r}}}}(\zeta),$
$R(\zeta),$ $\bar R(\zeta).$ Hence, in solving
the inverse scattering problem for \eqref{1.6}, instead of using all the scattering coefficients, it is sufficient to use as input a scattering data set containing the right scattering coefficients but not left scattering coefficients. 

\begin{theorem}
\label{theorem2.8}
Assume that the potentials $q$ and $r$ in \eqref{1.6} belong to the Schwartz class $\mathcal S(\mathbb R)$ in $x\in\mathbb R.$
The left scattering coefficients $T_{\text{\rm{l}}}(\zeta),$ $\bar T_{\text{\rm{l}}}(\zeta),$ $L(\zeta),$ $\bar L(\zeta)$ for \eqref{1.6} 
can be expressed in terms of the right scattering coefficients $T_{\text{\rm{r}}}(\zeta),$ $\bar T_{\text{\rm{r}}}(\zeta),$ $R(\zeta),$
$\bar R(\zeta)$ for \eqref{1.6} as  
\begin{equation}\label{2.57}
T_{\text{\rm{l}}}(\zeta)=e^{i\mu/2}T_{\text{\rm{r}}}(\zeta), \quad \bar{T_{\text{\rm{l}}}}(\zeta)=e^{i\mu/2}\bar{T_{\text{\rm{r}}}}(\zeta),
\end{equation}
\begin{equation}\label{2.58}
L(\zeta)=-e^{i\mu/2}\bar R(\zeta)\,\displaystyle\frac{T_{\text{\rm{r}}}(\zeta)}{\bar{T_{\text{\rm{r}}}}(\zeta)}, \quad
\bar L(\zeta)=-R(\zeta)\,\displaystyle\frac{\bar{T_{\text{\rm{r}}}}(\zeta)}{T_{\text{\rm{r}}}(\zeta)},
\end{equation}
where $\mu$ is the complex constant defined in \eqref{1.15}.
\end{theorem}

\begin{proof}
We get the first equality in \eqref{2.57} directly from the first equalities of \eqref{2.39} and \eqref{2.40}. Similarly, the second equality in 
\eqref{2.57} is obtained by using the second equalities in \eqref{2.39} and \eqref{2.40}. The analogs of the first and the second 
equalities in \eqref{2.58} for the linear system \eqref{1.11} are known from (2.45) of \cite{AEU2023a} and we have
\begin{equation}\label{2.59}
\tilde L(\zeta)=-e^{i\mu/2}\tilde{\bar R}(\zeta)\,\displaystyle\frac{\tilde T_{\text{\rm{r}}}(\zeta)}{\tilde{\bar{T_{\text{\rm{r}}}}}(\zeta)}, \quad
\tilde{\bar L}(\zeta)=-\tilde R(\zeta)\,\displaystyle\frac{\tilde{\bar{T_{\text{\rm{r}}}}}(\zeta)}{\tilde{T_{\text{\rm{r}}}}(\zeta)}.
\end{equation}
By using the relationships between the reflection coefficients for \eqref{1.6} and the reflection coefficients for \eqref{1.11} in 
\eqref{2.59}, namely, using \eqref{2.59} in \eqref{2.42}, we get the first and the second equalities in \eqref{2.58}, respectively.
\end{proof}

\section{The bound states}
\label{section3}

We recall that the potential pair $(q,r)$ appearing in \eqref{1.6} is assumed to belong to the Schwartz class $\mathcal S(\mathbb R)$
in $x\in\mathbb R.$ The bound states for \eqref{1.6} correspond to square integrable column-vector solutions to \eqref{1.6}. When the
spectral parameter $\zeta$ is real, by using any two linearly independent column-vector solutions to \eqref{1.6}, it is impossible to form
a square integrable solution to \eqref{1.6}. Hence, there are no bound states for \eqref{1.6} when $\zeta\in\mathbb R.$ A bound state
for \eqref{1.6} can only occur at a nonreal complex value of $\zeta.$

As indicated in Theorem~\ref{theorem2.2}, the potential pair $(\tilde q,\tilde r)$ in \eqref{1.11} belongs to the Schwartz class
$\mathcal S(\mathbb R)$ when the potential pair $(q,r)$ in \eqref{1.6} belongs $\mathcal S(\mathbb R).$ The bound states for
\eqref{1.6} are closely related to the meromorphic properties of the transmission coefficients in the complex $\zeta$-plane. From
\eqref{2.40} and \eqref{2.41} we know that the transmission coefficients for \eqref{1.6} and the transmission coefficients for \eqref{1.11} 
have similar meromorphic properties in the complex $\zeta$-plane. Thus, by using the facts outlined in Section~3 of
\cite{AEU2023a} for the bound states for \eqref{1.11}, we obtain the facts related to the bound states for \eqref{1.6}. We refer the 
reader to Section~3 of \cite{AEU2023a} for the details about the bound states for \eqref{1.11}, and in the following we provide a summary related to the bound states for 
\eqref{1.6} when the potential pair $(q,r)$ belongs to the Schwartz class.

\begin{enumerate}

\item[\text{\rm(a)}] As already indicated, the bound states for \eqref{1.6} cannot occur when $\zeta\in\mathbb R.$
A bound state can only occur at a nonreal
complex $\zeta$-value at which the transmission coefficient $T_{\text{\rm{r}}}(\zeta)$ has a pole in the first or third quadrant in the
complex $\zeta$-plane or the transmission coefficient $\bar T_{\text{\rm{r}}}(\zeta)$ has a pole in the second or fourth quadrant. 
As a consequence of \eqref{2.57}, the poles and their multiplicities for $T_{\text{\rm{l}}}(\zeta)$ and
$T_{\text{\rm{r}}}(\zeta)$ coincide and the poles and their multiplicities for $\bar T_{\text{\rm{l}}}(\zeta)$ and
$\bar T_{\text{\rm{r}}}(\zeta)$ coincide. From Theorem~\ref{theorem2.5}(b) we know that the transmission coefficients
$T_{\text{\rm{r}}}(\zeta)$ and $\bar T_{\text{\rm{r}}}(\zeta)$ are even in $\zeta,$ and hence the $\zeta$-values corresponding to the 
bound states for \eqref{1.6} are located symmetrically with respect to the origin of the complex $\zeta$-plane. Thus, it is convenient to describe the
bound-state poles of $T_{\text{\rm{r}}}(\zeta)$ and $\bar T_{\text{\rm{r}}}(\zeta)$ in terms of the parameter $\lambda$ related to $\zeta$ as in \eqref{2.7}. 

\item[\text{\rm(b)}] The number of poles of
$T_{\text{\rm{r}}}(\zeta)$ in the upper-half complex $\lambda$-plane is finite, and we use $\lambda_j$ for $1\le j\le N$ to denote those
distinct poles of $T_{\text{\rm{r}}}(\zeta).$ Similarly, the number of poles of $\bar T_{\text{\rm{r}}}(\zeta)$ in the lower-half complex
$\lambda$-plane is finite, and we use $\bar \lambda_j$ for $1\le j\le \bar N$ to denote those distinct $\bar \lambda_j$
values. It is possible that $T_{\text{\rm{r}}}(\zeta)$ has no poles in the upper-half complex $\lambda$-plane, in which case we have
$N=0.$ Similarly, it is possible that $\bar T_{\text{\rm{r}}}(\zeta)$ has no poles in the lower-half complex $\lambda$-plane, in which 
case we have $\bar N=0.$ The multiplicity of the pole of $T_{\text{\rm{r}}}(\zeta)$ at $\lambda=\lambda_j$ is finite, and we use $m_j$ to denote that multiplicity. 
Similarly, the multiplicity of the pole of $\bar T_{\text{\rm{r}}}(\zeta)$ at $\lambda=\bar\lambda_j$ is finite, and we use $\bar m_j$ to denote that multiplicity.

\item[\text{\rm(c)}] The bound-state information for \eqref{1.6} contains the sets $\left\{\lambda_j,m_j\right\}_{j=1}^N$ and
$\left\{\bar\lambda_j,\bar m_j\right\}_{j=1}^{\bar N}.$ For each bound state and multiplicity, we specify a bound-state normalization 
constant. To denote the bound-state normalization constants, we use the double-indexed quantities $c_{jk}$ for $1\le j\le N$ and
$0\le k\le m_j-1$ and the double-indexed quantities $\bar c_{jk}$ for $1\le j\le \bar N$ and
$0\le k\le \bar m_j-1.$ The construction of the complex-valued normalization constants $c_{jk}$ and $\bar c_{jk}$ for \eqref{1.6} is similar to the construction given in
\cite{AE2022,AEU2023a} of the bound-state normalization constants for \eqref{1.11}. Thus, the bound-state information for \eqref{1.6} consists of the two sets given by
\begin{equation}
\label{3.1}
\left\{\lambda_j,m_j,\{c_{jk}\}_{k=0}^{m_j-1}\right\}_{j=1}^N,\quad
\left\{\bar\lambda_j,\bar m_j,\{\bar c_{jk}\}_{k=0}^{\bar m_j-1}\right\}_{j=1}^{\bar N}.
\end{equation}
We refer the reader to Examples~6.1 and 6.2 in \cite{AEU2023a}, where it is illustrated how the bound-state normalization 
constants for \eqref{1.11} constructed by using the transmission coefficients and the bound-state dependency constants for \eqref{1.11}.

\item[\text{\rm(d)}] The bound-state information presented in \eqref{3.1} can be organized by using a pair of matrix triplets. We use the 
matrix triplet $(A,B,C)$ and the matrix triplet $(\bar A,\bar B,\bar C)$ to represent the information contained in the first and the second 
sets, respectively, appearing in \eqref{3.1}. The specification of the bound-state information in the form of a pair of matrix triplets is especially convenient in the solution of 
the inverse scattering problem for \eqref{1.6} with the help of the solution to a Marchenko system of integral equations. The advantage 
of using matrix triplets to represent the bound-state information is that it allows us to deal with any number of bound states with any 
multiplicities as if we deal only with a single bound state having the multiplicity equal to 1.
\end{enumerate}

The translation of the bound-state information from \eqref{3.1} to the matrix triplets $(A,B,C)$ and
$(\bar A,\bar B,\bar C)$ is as follows. For each bound state at $\lambda=\lambda_j$ for $1\le j\le N$ with the multiplicity $m_j,$ we form the matrix 
subtriplet $(A_j,B_j,C_j)$ by letting
\begin{equation}\label{3.2}
A_j:=\begin{bmatrix}
\lambda_j&1&0&\cdots&0&0\\
0&\lambda_j&1&\cdots&0&0\\
0&0&\lambda_j&\cdots&0&0\\
\vdots&\vdots&\vdots&\ddots&\vdots&\vdots\\
0&0&0&\cdots&\lambda_j&1\\
0&0&0&\dots&0&\lambda_j
\end{bmatrix},
\end{equation}
\begin{equation}\label{3.3}
B_j:=\begin{bmatrix}
0\\ \vdots \\
0\\
1
\end{bmatrix},\quad C_j:=\begin{bmatrix}
c_{j(m_j-1)}&c_{j(m_j-2)}&\cdots&c_{j1}&c_{j0}
\end{bmatrix}.
\end{equation}
Here, $A_j$ is the $m_j\times m_j$ square matrix in the Jordan canonical form with $\lambda_j$ in the diagonal entries,
$B_j$ is the column vector with $m_j$ components that are all zero with the exception of the last entry being $1,$
and $C_j$ is the row vector with $m_j$ components containing the bound-state normalization constants
$c_{jk}$ for $0\le k\le m_j-1$ in the order indicated in \eqref{3.3}.
In a similar manner, for each bound state at $\lambda=\bar\lambda_j$ for $1\le j \le\bar N,$ we form the matrix 
subtriplet $(\bar A_j,\bar B_j,\bar C_j)$ as
\begin{equation}\label{3.4}
\bar A_j:=\begin{bmatrix}
\bar\lambda_j&1&0&\cdots&0&0\\
0&\bar\lambda_j&1&\cdots&0&0\\
0&0&\bar\lambda_j&\cdots&0&0\\
\vdots&\vdots&\vdots&\ddots&\vdots&\vdots\\
0&0&0&\cdots&\bar\lambda_j&1\\
0&0&0&\dots&0&\bar\lambda_j
\end{bmatrix},
\end{equation}
\begin{equation}\label{3.5}
\bar B_j:=\begin{bmatrix}
0\\ \vdots \\
0\\
1
\end{bmatrix},\quad \bar C_j:=\begin{bmatrix}
\bar c_{j(\bar m_j-1)}&\bar c_{j(\bar m_j-2)}&\cdots&\bar c_{j1}&\bar c_{j0}
\end{bmatrix}.
\end{equation}
We note that $\bar A_j$ is the $\bar m_j\times \bar m_j$ square matrix in the Jordan canonical form with $\bar \lambda_j$ in the diagonal 
entries, $\bar B_j$ is the column vector with the first $\bar m_j-1$ entries being zero and with the last entry being $1,$
and $\bar C_j$ is the row vector with $\bar m_j$ components containing the bound-state normalization constants
$\bar c_{jk}$ for $0\le k\le\bar m_j-1$ in the order indicated in \eqref{3.5}.

By using the subtriplets $(A_j,B_j,C_j)$ and $(\bar A_j,\bar B_j,\bar C_j)$ given in \eqref{3.2}--\eqref{3.5}, we transform the bound-state information 
from \eqref{3.1} to the matrix triplets $(A,B,C)$ and $(\bar A,\bar B,\bar C)$ by letting
\begin{equation}\label{3.6}
A:=\begin{bmatrix}
A_1&0&\cdots&0&0\\
0&A_2&\cdots&0&0\\
\vdots&\vdots&\ddots&\vdots&\vdots\\
0&0&\cdots&A_{N-1}&0\\
0&0&\cdots&0&A_N
\end{bmatrix},
\quad
\bar A:=\begin{bmatrix}
\bar A_1&0&\cdots&0&0\\
0&\bar A_2&\cdots&0&0\\
\vdots&\vdots&\ddots&\vdots&\vdots\\
0&0&\cdots&\bar A_{\bar N-1}&0\\
0&0&\cdots&0&\bar A_{\bar N}
\end{bmatrix},
\end{equation}
\begin{equation}\label{3.7}
B:=\begin{bmatrix}
B_1\\
B_2\\
\vdots\\
B_N
\end{bmatrix},\quad \bar B:=\begin{bmatrix}
\bar B_1\\
\bar B_2\\
\vdots\\
\bar B_{\bar N}
\end{bmatrix},
\end{equation}
\begin{equation}\label{3.8}
C:=\begin{bmatrix}
C_1&C_2&\cdots&C_N
\end{bmatrix},
\quad 
\bar C:=\begin{bmatrix}
\bar C_1&\bar C_2&\cdots&\bar C_{\bar N}
\end{bmatrix}.
\end{equation}
We remark that the matrices $A,$ $B,$ $C,$ $\bar A,$ $\bar B,$ $\bar C$ each are block matrices, and the zeros in \eqref{3.6} denote the zero matrices of appropriate matrix sizes.
The matrix size of $A$ is $\mathcal N\times \mathcal N$ and the matrix size of $\bar A$ is $\bar{\mathcal N}\times \bar{\mathcal N},$ where we have defined
\begin{equation}\label{3.9}
\mathcal N:=\displaystyle\sum_{j=1}^{N} m_j, \quad \bar{\mathcal N}:=\displaystyle\sum_{j=1}^{\bar N} \bar m_j.
\end{equation}
The matrices $B$ and $\bar B$ are column vectors with $\mathcal N$ and $\bar{\mathcal N}$ components, respectively. Similarly, the matrices $C$ and $\bar C$
are row vectors with $\mathcal N$ and $\bar{\mathcal N}$ components, respectively.

\section{The Marchenko method}
\label{section4}

In this section we present the Marchenko method \cite{AM1963,M1986} for \eqref{1.6} by deriving the corresponding
Marchenko system of
linear integral equations. In the Marchenko method, the potentials $q$ and $r$ are recovered 
from the input scattering data set consisting of the scattering coefficients and the bound-state
information. The input is used to construct the kernel in the Marchenko system of linear integral equations
as well as the nonhomogeneous term in the Marchenko system. The potentials and all other relevant
quantities associated with \eqref{1.6} are then recovered from the solution to the Marchenko system.

In the next theorem, we present the derivation of the Marchenko system of integral equations for \eqref{1.6} in the absence of bound states.

\begin{theorem}
\label{theorem4.1}
Assume that the potentials $q$ and $r$ in \eqref{1.6} belong to the Schwartz class $\mathcal S(\mathbb R)$ in
$x\in\mathbb R.$ In the absence of bound states, the Marchenko system of linear integral equations for
\eqref{1.6} is given by
\begin{equation}\label{4.1}
\begin{split}
\begin{bmatrix}
0&0\\ 
\noalign{\medskip}
0&0
\end{bmatrix}=&\begin{bmatrix}
\bar K_1(x,y)&K_1(x,y)\\ \noalign{\medskip}\bar K_2(x,y)&K_2(x,y)
\end{bmatrix}+ \begin{bmatrix}
0&\hat{\bar R}(x+y)\\ 
\noalign{\medskip}
\hat R(x+y)&0
\end{bmatrix}\\
\noalign{\medskip}
&+\displaystyle\int_x^\infty dz \begin{bmatrix}
-i K_1(x,z)\,\hat R'(z+y)&\bar K_1(x,z)\,\hat{\bar R}(z+y)\\ 
\noalign{\medskip}
K_2(x,z)\,\hat R(z+y)&i\bar K_2(x,z)\,\hat{\bar R}'(z+y)
\end{bmatrix},\qquad x<y,
\end{split}
\end{equation}
where $\hat R(y)$ and $\hat{\bar R}(y)$ are related to the reflection coefficients $R(\zeta)$ and $\bar R(\zeta)$ for \eqref{1.6} 
via the Fourier transforms given by
\begin{equation}\label{4.2}
\hat R(y):=\displaystyle\frac{1}{2\pi}\displaystyle\int_{-\infty}^\infty  
d\lambda\,\displaystyle\frac{R(\zeta)}{\zeta}\,e^{i\lambda y},\quad \hat{\bar R}(y):=\displaystyle\frac{1}{2\pi}
\displaystyle\int_{-\infty}^\infty  d\lambda\,\displaystyle\frac{\bar R(\zeta)}{\zeta}\,e^{-i\lambda y},
\end{equation}
with $\hat R'(y)$ and $\hat{\bar R}'(y)$ denoting the derivatives of $\hat R(y)$ and $\hat{\bar R}(y),$ respectively, and
$\lambda$ being the parameter related to $\zeta$ as in \eqref{2.7}. The quantities $K_1(x,y),$ $K_2(x,y),$ $\bar K_1(x,y),$ and $\bar K_2(x,y)$ are related to
the components of the Jost solutions $\psi(\zeta,x)$ and $\bar \psi(\zeta,x)$ appearing in \eqref{2.1} as
\begin{equation}\label{4.3}
K_1(x,y):= 
\displaystyle\frac{1}{2\pi }\int_{-\infty}^\infty d\lambda \left[\displaystyle\frac{e^{i\mu/2}\,\,\psi_1(\zeta,x)}{\zeta\,E(x)}\right] e^{-i\lambda y},
\end{equation}
\begin{equation}\label{4.4}
K_2(x,y):= 
\displaystyle\frac{1}{2\pi }\int_{-\infty}^\infty d\lambda \left[\psi_2(\zeta,x)-e^{i\lambda x}\right] e^{-i\lambda y},
\end{equation}
\begin{equation}\label{4.5}
\bar K_1(x,y):= 
\displaystyle\frac{1}{2\pi }\int_{-\infty}^\infty 
d\lambda \left[\displaystyle\frac{e^{i\mu/2}}{E(x)}\bar\psi_1(\zeta,x)-e^{-i\lambda x}\right] e^{i\lambda y},
\end{equation}
\begin{equation}\label{4.6}
\bar K_2(x,y):= 
\displaystyle\frac{1}{2\pi }\int_{-\infty}^\infty d\lambda
\left[\displaystyle\frac{\bar\psi_2(\zeta,x)}{\zeta}\right] e^{i\lambda y},
\end{equation}
with $E(x)$ and $\mu$ being the quantities defined in \eqref{1.13} and \eqref{1.15}, 
respectively.
\end{theorem}

\begin{proof} For notational simplicity in the derivation of \eqref{4.1}, we suppress the arguments and write
$\psi$ for $\psi(\zeta,x),$ $\bar\psi$ for $\bar\psi(\zeta,x),$ $\phi$ for $\phi(\zeta,x),$ 
$\bar\phi$ for $\bar\phi(\zeta,x),$ $T_{\text{\rm{r}}}$ for $T_{\text{\rm{r}}}(\zeta),$ $\bar T_{\text{\rm{r}}}$ for $\bar T_{\text{\rm{r}}}(\zeta),$ $R$ for $R(\zeta),$
 $\bar R$ for $\bar R(\zeta),$ and $E$ for $E(x).$
From the asymptotics in \eqref{2.3} and \eqref{2.4} we see that
the column-vector Jost solutions $\psi$ and $\bar\psi$ to \eqref{1.6} are linearly independent, and hence those four
column-vector solutions form a fundamental set of column-vector solutions to \eqref{1.6}.
Thus, each of the other
two column-vector Jost solutions $\phi$ and $\bar\phi$ can be expressed as linear combinations of
$\psi$ and $\bar\psi.$ With the help of \eqref{2.3}, \eqref{2.4}, \eqref{2.33}, and \eqref{2.34}, for $\zeta\in\mathbb R$ we obtain
\begin{equation*}
\begin{cases}
\phi=\displaystyle\frac{1}{T_{\text{\rm{r}}}}\,\bar{ \psi}+\displaystyle\frac{R}{T_{\text{\rm{r}}}}\,\psi,
\\ \noalign{\medskip}
\bar\phi=\displaystyle\frac{\bar R}{\bar T_{\text{\rm{r}}}}\,\bar{\psi}+\displaystyle\frac{1}{\bar T_{\text{\rm{r}}}}\,\psi,
\end{cases}
\end{equation*}
or equivalently
\begin{equation}\label{4.8}
\begin{cases}
T_{\text{\rm{r}}}\,\phi=\bar{ \psi}+R\,\psi,
\\ \noalign{\medskip}
\bar T_{\text{\rm{r}}}\,\bar\phi=\bar R\,\bar{ \psi}+\psi,
\end{cases}
\end{equation}
which forms a basis for our Riemann--Hilbert problem. The solution to the 
Riemann--Hilbert problem consists of the construction of the Jost solutions from 
the knowledge of the scattering coefficients $T_{\text{\rm{r}}},$ $\bar T_{\text{\rm{r}}},$ $R,$ $\bar R$ appearing as coefficients in \eqref{4.8}.
We derive our Marchenko system of integral equations starting from \eqref{4.8} by proceeding as follows.
We first combine the two column-vector equations in \eqref{4.8} 
and obtain the $2\times 2$ matrix-valued system given by
\begin{equation}\label{4.9}
\begin{bmatrix}
T_{\text{\rm{r}}}\,\phi&\bar T_{\text{\rm{r}}}\,\bar\phi
\end{bmatrix}=\begin{bmatrix}
\bar\psi&\psi
\end{bmatrix}+\begin{bmatrix}
R\,\psi&\bar R\,\bar\psi
\end{bmatrix}.
\end{equation}
Using \eqref{2.1} and \eqref{2.2}, we write \eqref{4.9} as
\begin{equation}\label{4.10}
\begin{bmatrix}
T_{\text{\rm{r}}}\,\phi_1&\bar{T_{\text{\rm{r}}}}\,\bar\phi_1\\
\noalign{\medskip}
T_{\text{\rm{r}}}\,\phi_2&\bar{T_{\text{\rm{r}}}}\,\bar\phi_2
\end{bmatrix}=\begin{bmatrix}
\bar\psi_1&\psi_1
\\ \noalign{\medskip}
\bar\psi_2&\psi_2
\end{bmatrix}+\begin{bmatrix}
R\,\psi_1&\bar R\,\bar\psi_1
\\ \noalign{\medskip}
R\,\psi_2&\bar R\,\bar\psi_2
\end{bmatrix}.
\end{equation}
Premultiplying both sides of \eqref{4.10} by the $2\times 2$ matrix $\begin{bmatrix}
e^{i\mu/2}E^{-1}&0\\
0&1
\end{bmatrix},$ we obtain
\begin{equation}\label{4.11}
\begin{split}
\begin{bmatrix}
e^{i\mu/2}E^{-1}T_{\text{\rm{r}}}\,\phi_1&e^{i\mu/2}E^{-1}\bar{T_{\text{\rm{r}}}}\,\bar\phi_1\\
\noalign{\medskip}
T_{\text{\rm{r}}}\,\phi_2&\bar{T_{\text{\rm{r}}}}\,\bar\phi_2
\end{bmatrix}=&\begin{bmatrix}
e^{i\mu/2}E^{-1}\bar\psi_1&e^{i\mu/2}E^{-1}\psi_1
\\ \noalign{\medskip}
\bar\psi_2&\psi_2
\end{bmatrix}\\
&+\begin{bmatrix}
e^{i\mu/2}E^{-1}R\,\psi_1&e^{i\mu/2}E^{-1}\bar R\,\bar\psi_1
\\ \noalign{\medskip}
R\,\psi_2&\bar R\,\bar\psi_2
\end{bmatrix}.
\end{split}
\end{equation}
Subtracting the $2\times 2$ matrix $\begin{bmatrix}
e^{-i\lambda x}&0\\
0&e^{i\lambda x}
\end{bmatrix}$ from each side of \eqref{4.11} and then dividing the off-diagonal entries by $\zeta$ in the resulting matrix equality,
we obtain
\begin{equation}\label{4.12}
\begin{split}
\begin{bmatrix}
e^{i\mu/2}E^{-1}T_{\text{\rm{r}}}\,\phi_1-e^{-i\lambda x}&\displaystyle\frac{e^{i\mu/2}}{\zeta}E^{-1}\bar{T_{\text{\rm{r}}}}\,\bar\phi_1\\
\noalign{\medskip}
\displaystyle\frac{1}{\zeta}T_{\text{\rm{r}}}\,\phi_2&\bar{T_{\text{\rm{r}}}}\,\bar\phi_2-e^{i\lambda x}
\end{bmatrix}=&\begin{bmatrix}
e^{i\mu/2}E^{-1}\bar\psi_1-e^{-i\lambda x}&\displaystyle\frac{e^{i\mu/2}}{\zeta}E^{-1}\psi_1
\\ \noalign{\medskip}
\displaystyle\frac{1}{\zeta}\bar\psi_2&\psi_2-e^{i\lambda x}
\end{bmatrix}\\
&+\begin{bmatrix}
e^{i\mu/2}E^{-1}R\,\psi_1&\displaystyle\frac{e^{i\mu/2}}{\zeta}E^{-1}\bar R\,\bar\psi_1
\\ \noalign{\medskip}
\displaystyle\frac{1}{\zeta}R\,\psi_2&\bar R\,\bar\psi_2
\end{bmatrix}.
\end{split}
\end{equation}
Next, we take the Fourier transform of \eqref{4.12} 
with $\int_{-\infty}^\infty d\lambda\,e^{i\lambda y}/2\pi$  in the first columns and 
with $\int_{-\infty}^\infty d\lambda\,e^{-i\lambda y}/2\pi$ in the second columns. This yields
the $2\times 2$ matrix-valued equation
\begin{equation}
\label{4.13}
\text{\rm{LHS}}=\mathcal K(x,y)+\text{\rm{RHS}},
\end{equation}
where we have defined
\begin{equation}
\label{4.14}
\mathcal K(x,y):=\begin{bmatrix}
\bar K_1(x,y)&K_1(x,y)
\\
\noalign{\medskip}
\bar K_2(x,y)&K_2(x,y)
\end{bmatrix},
\end{equation}
with the entries $K_1(x,y),$ $K_2(x,y),$ $\bar K_1(x,y),$ and $\bar K_2(x,y)$ are as in
\eqref{4.3}--\eqref{4.6}, respectively. In \eqref{4.13} we have
\begin{equation}
\label{4.15}
\text{\rm{LHS}}:=\begin{bmatrix}
\text{\rm{LHS}}_{11}&\text{\rm{LHS}}_{12}
\\
\noalign{\medskip}
\text{\rm{LHS}}_{21}&\text{\rm{LHS}}_{22}
\end{bmatrix},
\end{equation}
\begin{equation}
\label{4.16}
\text{\rm{RHS}}:=\begin{bmatrix}
\text{\rm{RHS}}_{11}&\text{\rm{RHS}}_{12}
\\
\noalign{\medskip}
\text{\rm{RHS}}_{21}&\text{\rm{RHS}}_{22}
\end{bmatrix},
\end{equation}
with the matrix entries in \eqref{4.15} and \eqref{4.16} defined as
\begin{equation}
\label{4.17}
\text{\rm{LHS}}_{11}:=\displaystyle\int_{-\infty}^\infty \frac{d\lambda}{2\pi}\left[\displaystyle\frac{e^{i\mu/2}}{E(x)}\,T_{\text{\rm{r}}}(\zeta)\,\phi_1(\zeta,x)-e^{-i\lambda x}\right]e^{i\lambda y},
\end{equation}
\begin{equation}
\label{4.18}
\text{\rm{LHS}}_{12}:=
\displaystyle\int_{-\infty}^\infty \frac{d\lambda}{2\pi}\left[\displaystyle\frac{e^{i\mu/2}}{\zeta E(x)}\,\bar{T_{\text{\rm{r}}}}(\zeta)\,\bar\phi_1(\zeta,x)\right]e^{-i\lambda y},
\end{equation}
\begin{equation}
\label{4.19}
\text{\rm{LHS}}_{21}:=\displaystyle\int_{-\infty}^\infty \displaystyle\frac{d\lambda}{2\pi}\left[\displaystyle\frac{1}{\zeta}\,T_{\text{\rm{r}}}(\zeta)\,\phi_2(\zeta,x)\right]e^{i\lambda y},
\end{equation}
\begin{equation}
\label{4.20}
\text{\rm{LHS}}_{22}:=
\displaystyle\int_{-\infty}^\infty \frac{d\lambda}{2\pi}\left[\bar{T_{\text{\rm{r}}}}(\zeta)\,\bar{ \phi}_2(\zeta,x)-e^{i\lambda x}\right]e^{-i\lambda y},
\end{equation}
\begin{equation}
\label{4.21}
\text{\rm{RHS}}_{11}:=
\displaystyle\int_{-\infty}^\infty \frac{d\lambda}{2\pi}\left[\displaystyle\frac{e^{i\mu/2}}{E(x)}\,R(\zeta)\,\psi_1(\zeta,x)\right]
e^{i\lambda y},
\end{equation}
\begin{equation}
\label{4.22}
\text{\rm{RHS}}_{12}:=
\displaystyle\int_{-\infty}^\infty \displaystyle\frac{d\lambda}{2\pi}\left[\displaystyle\frac{e^{i\mu/2}}{\zeta E(x)}\,\bar R(\zeta)\,\bar\psi_1(\zeta,x)\right]e^{-i\lambda y},
\end{equation}
\begin{equation}
\label{4.23}
\text{\rm{RHS}}_{21}:=
\displaystyle\int_{-\infty}^\infty \displaystyle\frac{d\lambda}{2\pi}\left[\displaystyle\frac{1}{\zeta}\,R(\zeta)\,\psi_2(\zeta,x)\right]e^{i\lambda y},
\end{equation}
\begin{equation}
\label{4.24}
\text{\rm{RHS}}_{22}:=\displaystyle\int_{-\infty}^\infty \frac{d\lambda}{2\pi}\left[\bar R(\zeta)\,\bar\psi_2(\zeta,x)\right]e^{-i\lambda y}.
\end{equation}
In the absence of bound states, when $x<y$ the integrands in \eqref{4.3} and \eqref{4.4} are
analytic in $\lambda\in\mathbb C^+,$ are continuous in $\lambda\in\overline{\mathbb C^+},$ and behave as $e^{i\lambda(y-x)}O(1/\lambda)$ as $\lambda\to\infty$ in
$\overline{\mathbb C^+}.$ Similarly,
in the absence of bound states, when $x<y$ the integrands in \eqref{4.5} and \eqref{4.6} are
analytic in $\lambda\in\mathbb C^-,$ are continuous in $\lambda\in\overline{\mathbb C^-},$ and behave as $e^{-i\lambda(y-x)}O(1/\lambda)$ as $\lambda\to\infty$ in
$\overline{\mathbb C^-}.$
Thus, from Jordan's lemma it follows that the matrix $\mathcal K(x,y)$ defined in \eqref{4.12} is equal to zero when $x>y.$
On the other hand, in the absence of bound states, when $x<y$ with the help of
Theorems~\ref{theorem2.1}, \ref{theorem2.4}, and \ref{theorem2.6} we observe that the integrands in \eqref{4.17} and \eqref{4.19} are analytic
in $\lambda\in\mathbb C^+,$ are continuous in $\lambda\in\overline{\mathbb C^+},$
and behave uniformly as $O(1/\lambda)$ as $\lambda\to\infty$ in $\overline{\mathbb C^+}.$
Similarly, when $x<y,$ with the help of
Theorems~\ref{theorem2.1}, \ref{theorem2.4}, and \ref{theorem2.6}, we observe that the integrands in \eqref{4.18} and \eqref{4.20} are analytic
in $\lambda\in\mathbb C^-,$ continuous in $\lambda\in\overline{\mathbb C^-},$
and behave as $O(1/\lambda)$ as $\lambda\to\infty$ in $\overline{\mathbb C^-}.$
Hence, when $x<y,$ using Jordan's lemma and the residue theorem, we conclude that the matrix
$\text{\rm{LHS}}$ defined in \eqref{4.15} is zero.
In fact, from the continuity of the Jost solutions stated in Theorem~\ref{theorem2.1},
the continuity and asymptotic properties of the scattering coefficients stated in Theorem~\ref{theorem2.7}, 
the large $\zeta$-asymptotics of the
Jost solutions stated in Theorems~\ref{theorem2.4},
we observe that each integrand in \eqref{4.3}--\eqref{4.6} and \eqref{4.17}--\eqref{4.24} is continuous in $\lambda\in\mathbb R$ and 
behaves as $O(1/\lambda)$ as 
$\lambda\to\pm\infty.$ Thus, the $L^2$-Fourier transforms in \eqref{4.3}--\eqref{4.6} and \eqref{4.17}--\eqref{4.24} are all
well defined. Using the inverse Fourier transform, from \eqref{4.3}--\eqref{4.6} we get
\begin{equation}\label{4.25}
\displaystyle\frac{e^{i\mu/2}}{\zeta E(x)}\,\psi_1(\zeta,x)=\displaystyle\int_x^\infty dy\,K_1(x,y)\,e^{i\lambda y},
\end{equation}
\begin{equation}\label{4.26}
\psi_2(\zeta,x)=e^{i\lambda x}+\displaystyle\int_x^\infty dy\,K_2(x,y)\,e^{i\lambda y},
\end{equation}
\begin{equation}\label{4.27}
\displaystyle\frac{e^{i\mu/2}}{E(x)}\,\bar\psi_1(\zeta,x)=e^{-i\lambda x}+\displaystyle\int_x^\infty dy\,\bar K_1(x,y)\,e^{-i\lambda y},
\end{equation}
\begin{equation}\label{4.28}
\displaystyle\frac{1}{\zeta}\,\bar\psi_2(\zeta,x)=\displaystyle\int_x^\infty dy\,\bar K_2(x,y)\,e^{-i\lambda y}.
\end{equation}
Form \eqref{4.2}, by using the inverse Fourier transform we have
\begin{equation}\label{4.29}
\displaystyle\frac{R(\zeta)}{\zeta}=\displaystyle\int_{-\infty}^\infty ds\,\hat R(s)\,e^{-i\lambda s},
\quad 
\frac{\bar R(\zeta)}{\zeta}=\displaystyle\int_{-\infty}^\infty ds\,\hat{\bar R}(s)\,e^{i\lambda s}.
\end{equation}
Taking the $y$-derivatives, from \eqref{4.2} we obtain
\begin{equation}\label{4.30}
\hat R'(y)=\displaystyle\frac{i}{2\pi }\displaystyle\int_{-\infty}^\infty  
d\lambda\,\displaystyle\frac{R(\zeta)}{\zeta}\lambda\,e^{i\lambda y},
\quad \hat{\bar R}'(y)=-\displaystyle\frac{i}{2\pi }\displaystyle\int_{-\infty}^\infty  
d\lambda\,\displaystyle\frac{\bar R(\zeta)}{\zeta}\lambda\,e^{-i\lambda y}.
\end{equation}
Using the inverse Fourier transform, from \eqref{4.30} we get
\begin{equation}\label{4.31}
\displaystyle\frac{R(\zeta)}{\zeta}\,\lambda=-i\displaystyle\int_{-\infty}^\infty ds\,\hat R'(s)\,e^{-i\lambda s},
\quad 
\frac{\bar R(\zeta)}{\zeta}\,\lambda=i\displaystyle\int_{-\infty}^\infty ds\,\hat{\bar R}'(s)\,e^{i\lambda s}.
\end{equation}
Next, we take the Fourier transform of each component of the matrix $\text{\rm{RHS}}$ appearing in \eqref{4.16}. For this, we first write \eqref{4.21} in an equivalent form, i.e. 
\begin{equation}
\label{4.32}
\text{\rm{RHS}}_{11}=
\displaystyle\int_{-\infty}^\infty \frac{d\lambda}{2\pi}\,e^{i\lambda y}\left(\displaystyle\frac{e^{i\mu/2}}{\zeta E(x)}\psi_1(\zeta,x)\right)\left(\frac{R(\zeta)}{\zeta}\,\lambda\right).
\end{equation}
Then, using \eqref{4.25} and the first equality of \eqref{4.31}
on the right-hand side of \eqref{4.32}, we get
\begin{equation}
\label{4.33}
\text{\rm{RHS}}_{11}=-i\displaystyle\int_x^\infty  dz\,K_1(x,z)\,\hat R'(z+y),
\end{equation}
where we have used the fact that $K_1(x,z)$ vanishes when $x>z.$ Similarly, we write \eqref{4.24} in an equivalent form as
\begin{equation}
\label{4.34}
\text{\rm{RHS}}_{22}=\displaystyle\int_{-\infty}^\infty \frac{d\lambda}{2\pi}
\,e^{-i\lambda y}\left(e^{-i\mu/2}\,E(x)\,\displaystyle\frac{\bar{\psi}_2(\zeta,x)}{\zeta}\right)\left(\frac{\bar R(\zeta)}{\zeta}\,\lambda\right).
\end{equation}
Then, using \eqref{4.28} and the second equality of \eqref{4.31}
on the right-hand side of \eqref{4.34}, we obtain
\begin{equation}\label{4.35}
\text{\rm{RHS}}_{22}=i\displaystyle\int_x^\infty  dz\,\bar K_2(x,z)\,\hat{\bar R}'(z+y),
\end{equation}
where we have used the fact that $\bar K_2(x,z)$ vanishes when $x>z.$
Proceeding in a similar manner, by using \eqref{4.26}, \eqref{4.27}, and \eqref{4.29}, we write
\eqref{4.22} and \eqref{4.23}, respectively, as
\begin{equation}\label{4.36}
\text{\rm{RHS}}_{12}
=\hat{\bar R}(x+y)+\displaystyle\int_x^\infty  dz\,\bar K_1(x,z)\,\hat{\bar R}(z+y),
\end{equation}
\begin{equation}\label{4.37}
\text{\rm{RHS}}_{21}=\hat R(x+y)+\displaystyle\int_x^\infty  dz\,K_2(x,z)\,\hat R(z+y).
\end{equation}
Hence, using \eqref{4.33}, \eqref{4.35}, \eqref{4.36}, and \eqref{4.37}
in \eqref{4.13}, we observe that $\text{\rm{RHS}}$ is
equal to the sum of the second and third terms on the right-hand side of \eqref{4.1}.
Thus, the proof is complete.
\end{proof}

When \eqref{1.6} has bound states,
the only modification needed in the proof of Theorem~\ref{theorem4.1} is that
the quantity $\text{\rm{LHS}}$ appearing in \eqref{4.13} is no longer equal to the zero matrix
due to the fact that we must take into
account the bound-state poles of the transmission coefficients $T_{\text{\rm{r}}}(\zeta)$ and $\bar T_{\text{\rm{r}}}(\zeta)$ in evaluating
the integrals in \eqref{4.17}--\eqref{4.20}. Those integrals, after using the poles of $T_{\text{\rm{r}}}(\zeta)$ and $\bar T_{\text{\rm{r}}}(\zeta)$ and
the bound-state dependency constants for \eqref{1.6}, can be explicitly evaluated and the results can be expressed in terms of the matrix triplets
$(A,B,C)$ and $(\bar A,\bar B,\bar C)$ appearing in \eqref{3.6}--\eqref{3.8}. This yields the Marchenko system of integral equations presented in
the next theorem in the presence of bound states for \eqref{1.6}. The proof of the theorem is analogous to the proof of Theorem~4.2 of \cite{AEU2023b}
involving the derivation of the Marchenko system for \eqref{1.11} in the presence of bound states. Hence, we omit the proof.
In order to present the Marchenko system of integral equations for \eqref{1.6} in the presence of bound states, we introduce the $2\times 2$ matrix-valued
quantities $\Omega(y)$ and $\bar\Omega(y)$ as
\begin{equation}\label{4.38}
\Omega(y):=\hat R(y)+C\,e^{iAy}B,\quad \bar\Omega(y):=\hat{\bar R}(y)+\bar C\,e^{-i\bar A y}\bar B,
\end{equation}
where we use $e^{iAy}$ and $e^{-i\bar A y}$ to denote the corresponding matrix exponentials. By taking the $y$-derivative of the two matrix equalities in \eqref{4.38}, we obtain
\begin{equation}
\label{4.39}
\Omega'(y)=\hat R'(y)+i \,C A\, e^{iAy} B,\quad \bar\Omega'(y)=\hat{\bar R}'(y)-i\, \bar C \bar A \,e^{-i\bar A y} \bar B.
\end{equation}

\begin{theorem}
\label{theorem4.2}
Assume that the potentials $q$ and $r$ in \eqref{1.6} belong to the Schwartz class $\mathcal S(\mathbb R)$ in $x\in\mathbb R.$
Let $(A,B,C)$ and $(\bar A,\bar B,\bar C)$ be the pair of matrix triplets representing the bound-state information
for \eqref{1.6}.
In the presence of bound states, the Marchenko system for \eqref{1.6} is given by
\begin{equation}\label{4.40}
\begin{split}
\begin{bmatrix}
0&0\\ \noalign{\medskip}0&0
\end{bmatrix}=&\begin{bmatrix}
\bar K_1(x,y)&K_1(x,y)\\ \noalign{\medskip}\bar K_2(x,y)&K_2(x,y)
\end{bmatrix}+ \begin{bmatrix}
0&\bar\Omega(x+y)\\ \noalign{\medskip}\Omega(x+y)&0
\end{bmatrix}\\
\noalign{\medskip}
&+\displaystyle\int_x^\infty dz\begin{bmatrix}
-iK_1(x,z)\,\Omega'(z+y)&\bar K_1(x,z)\,\bar\Omega(z+y)\\ \noalign{\medskip}
K_2(x,z)\,\Omega(z+y)&i\bar K_2(x,z)\,\bar\Omega'(z+y)
\end{bmatrix},\qquad x<y,
\end{split}
\end{equation}
where $\Omega(y),$ $\bar\Omega(y),$ $\Omega'(y),$ $\bar\Omega'(y)$ are the quantities
appearing in \eqref{4.38} and \eqref{4.39}, and $K_1(x,y),$ $K_2(x,y),$ $\bar K_1(x,y),$ $\bar K_2(x,y)$ are the quantities appearing in
\eqref{4.3}--\eqref{4.6}, respectively.
\end{theorem}

Using the four equalities arising from the four entries in \eqref{4.40}, we write the Marchenko system as a coupled system of four 
integral equations holding for $x<y$ as
\begin{equation}\label{4.41}
\begin{cases}
\bar K_1(x,y)-i\displaystyle\int_x^\infty dz\,K_1(x,z)\,\Omega'(z+y)=0,\\
\noalign{\medskip}
K_1(x,y)+\bar\Omega(x+y)+\displaystyle\int_x^\infty dz\,\bar K_1(x,z)\,\bar\Omega(z+y)=0,\\
\noalign{\medskip}
\bar K_2(x,y)+\Omega(x+y)+\displaystyle\int_x^\infty dz\, K_2(x,z)\,\Omega(z+y)=0,\\
\noalign{\medskip}
K_2(x,y)+i\displaystyle\int_x^\infty dz\,\bar K_2(x,z)\,\bar\Omega'(z+y)=0.
\end{cases}
\end{equation}
We can uncouple the Marchenko system \eqref{4.41} by using the first line in the second equality and by using the fourth line in the 
third equality. We get
\begin{equation}\label{4.42}
\begin{cases}
K_1(x,y)+\bar\Omega(x+y)+i\displaystyle\int_x^\infty dz\, K_1(x,z) \int_x^\infty 
ds\,\Omega'(z+s)\,\bar\Omega(s+y)=0,
\\
\noalign{\medskip}
\bar K_2(x,y)+\Omega(x+y)-i\displaystyle\int_x^\infty dz\,\bar K_2(x,z) \int_x^\infty 
ds\,\bar\Omega'(z+s)\,\Omega(s+y)=0,
\end{cases}
\end{equation}
\begin{equation}\label{4.43}
\begin{cases}
\bar K_1(x,y)=i\displaystyle\int_x^\infty dz\,K_1(x,z)\,\Omega'(z+y),
\\ \noalign{\medskip}
K_2(x,y)=-i\displaystyle\int_x^\infty dz\,\bar K_2(x,z)\,\bar\Omega'(z+y),
\end{cases}
\end{equation}
where it is understood that we first solve the two uncoupled integral equations given in \eqref{4.42} and obtain $K_1(x,y)$ and
$\bar K_2(x,y)$ and then use those values in the integrands in \eqref{4.43} and recover $\bar K_1(x,y)$ and $K_2(x,y).$

In the next theorem, we relate the quantities $K_1(x,x),$ $\bar K_1(x,x),$ $K_2(x,x),$ $\bar K_2(x,x)$ obtained from the solution $\mathcal K(x,y)$ to the
Marchenko system \eqref{4.40} to some quantities related to the potential pair $(q,r)$ in \eqref{1.6}. 

\begin{theorem}\label{theorem4.3}
Suppose that the potentials $q$ and $r$  appearing in \eqref{1.6} belong to the Schwartz class $\mathcal S(\mathbb R)$ in $x\in\mathbb R.$
Let $ \mathcal K(x,y)$ be the solution to the Marchenko system
\eqref{4.40}, with the components $K_1(x,y),$ $K_2(x,y),$ $\bar K_1(x,y),$ $\bar K_2(x,y)$ 
as in \eqref{4.14}. In the limit $y\to x^+$ we have
\begin{equation}\label{4.44}
K_1(x,x)=-\displaystyle\frac{e^{i\mu/2}}{2}\displaystyle\frac{q(x)}{E(x)},
\end{equation}
\begin{equation}\label{4.45}
K_2(x,x)=-\displaystyle\frac{iq(x)\,r(x)}{4}-\displaystyle\frac{i}{4}\int_x^\infty dy\,q(y)\,r'(y),
\end{equation}
\begin{equation}\label{4.46}
\bar K_1(x,x)=\displaystyle\frac{1}{2}\int_x^\infty dy\,q(y)\,r'(y),
\end{equation}
\begin{equation}\label{4.47}
\bar K_2(x,x)=-\displaystyle\frac{e^{-i\mu/2}}{2}\,r(x)\,E(x),
\end{equation}
where $E(x)$ and $\mu$ are the quantities defined in \eqref{1.13} and \eqref{1.15}, respectively, and for notational simplicity we use
$K_1(x,x),$ $K_2(x,x),$ $\bar K_1(x,x),$ $\bar K_2(x,x)$ for $K_1(x,x^+),$ $K_2(x,x^+),$ $\bar K_1(x,x^+),$ $\bar K_2(x,x^+),$ respectively.
\end{theorem}

\begin{proof} We write \eqref{4.25} as 
\begin{equation}\label{4.48}
\displaystyle\frac{e^{i\mu/2}}{\zeta E(x)}\psi_1(\zeta,x)= 
\displaystyle\int_x^\infty dy\left[K_1(x,y)\,\frac{d}{dy}\,\frac{e^{i\lambda y}}{i\lambda}\right],
\end{equation}
where we recall that the parameter $\lambda$ is related to $\zeta$ as in \eqref{2.7}.
Using integration by parts in \eqref{4.48}, we obtain 
\begin{equation}
\label{4.49}
\displaystyle\frac{e^{i\mu/2}\,\psi_1(\zeta,x)}{\zeta\,E(x)}=
-K_1(x,x)\displaystyle\frac{e^{i\lambda x}}{i\lambda}-\int_x^\infty dy\,\frac{e^{i\lambda y}}
{i\lambda}\frac{\partial\,K_1(x,y)}{\partial y},
\end{equation}
where we have used $K_1(x,+\infty)=0.$ Letting $\lambda\to\infty$ in $\overline{\mathbb C^+},$ from \eqref{4.49} we get
\begin{equation}
\label{4.50}
\displaystyle\frac{e^{i\mu/2}\,\psi_1(\zeta,x)}{\zeta\,E(x)}=
-K_1(x,x)\displaystyle\frac{e^{i\lambda x}}{i\lambda}+O\left(\displaystyle\frac{1}{\lambda^2}\right).
\end{equation}
Using the large $\zeta$-asymptotics of $\psi_1(\zeta,x)$ given in \eqref{2.23} on the left-hand side of \eqref{4.50}, we have
\begin{equation}
\label{4.51}
\displaystyle\frac{e^{i\mu/2}\,e^{i\lambda x}}{E(x)}\left[\displaystyle\frac{q(x)}{2i\lambda} +O\left(\displaystyle\frac{1}{\lambda^2}\right)\right]=
-\displaystyle\frac{K_1(x,x)\,e^{i\lambda x}}{i\lambda}+O\left(\displaystyle\frac{1}{\lambda^2}\right), \qquad \lambda\to\infty  \text{\rm{ in }} \overline{\mathbb C^+}.
\end{equation}
A comparison of the $O(1/\lambda)$-terms on both sides of \eqref{4.51} yields \eqref{4.44}.
The equalities \eqref{4.45}--\eqref{4.47} are obtained in a similar manner.
\end{proof}

In the next theorem, we show how to recover the quantity $E(x),$ the potential pair $(q,r),$ and the Jost solutions
$\psi(\zeta,x)$ and $\bar\psi(\zeta,x)$ to \eqref{1.6} from the solution $\mathcal K(x,y)$ to the Marchenko system \eqref{4.40}. 

\begin{theorem}
\label{theorem4.4}
Assume that the potentials $q$ and $r$ in \eqref{1.6} belong to the Schwartz class $\mathcal S(\mathbb R)$ in 
$x\in\mathbb R.$ The quantity $E(x)$ in \eqref{1.13}, the constant $\mu$ in \eqref{1.15}, the potential pair $(q,r),$
and the Jost solutions $\psi(\zeta,x)$ and $\bar\psi(\zeta,x)$ to 
\eqref{1.6} are recovered from the solution $\mathcal K(x,y)$ to the Marchenko system \eqref{4.40} as follows:

\begin{enumerate}

\item[\text{\rm(a)}] We have
\begin{equation}\label{4.52}
E(x)=\exp\left(2i\displaystyle\int_{-\infty}^{x}dz\,P(z)\right), \quad \mu=4\displaystyle\int_{-\infty}^\infty dz\,P(z),
\end{equation}
where $P(x)$ is the scalar quantity constructed from $\bar K_1(x,y)$ and $K_2(x,y)$ as
\begin{equation}\label{4.53}
P(x):=K_1(x,x)\,\bar K_2(x,x).
\end{equation}

\item[\text{\rm(b)}] The potential pair $(q,r)$ is recovered as
\begin{equation}\label{4.54}
q(x)=-2K_1(x,x)\,\exp\left(-2i\displaystyle\int_x^{\infty} dz\,P(z)\right),
\end{equation}
\begin{equation}\label{4.55}
r(x)=-2\bar K_2(x,x)\,\exp\left(2i\displaystyle\int_x^{\infty} dz\,P(z)\right).
\end{equation}

\item[\text{\rm(c)}] The Jost solutions $\psi(\zeta,x)$ and $\bar\psi(\zeta,x)$ 
to \eqref{1.1} are recovered as
\begin{equation}\label{4.56}
\psi_1(\zeta,x)= \zeta\left(\displaystyle\int_x^\infty dy\,K_1(x,y)\,e^{i\zeta^2y}\right)\exp\left(-2i\displaystyle\int_x^{\infty} dz\,P(z)\right),
\end{equation}
\begin{equation}\label{4.57}
\psi_2(\zeta,x)=e^{i\zeta^2x}+\displaystyle\int_x^\infty dy\,K_2(x,y)\,e^{i\zeta^2 y},
\end{equation}
\begin{equation}\label{4.58}
\bar{\psi}_1(\zeta,x)=\left (e^{-i\zeta^2x}+\displaystyle\int_x^\infty dy\,\bar K_1(x,y)\,e^{-i\zeta^2 y}\right)\exp\left(-2i\displaystyle\int_x^{\infty} dz\,P(z)\right),
\end{equation}
\begin{equation}\label{4.59}
\bar{\psi}_2(\zeta,x)= \zeta \displaystyle\int_x^\infty dy\,\bar K_2(x,y)\,e^{-i\zeta^2y},
\end{equation}
where we recall that $\psi_1(\zeta,x),$ $\psi_2(\zeta,x),$ $\bar{\psi}_1(\zeta,x),$ and $\bar{\psi}_2(\zeta,x)$ 
are the components of the Jost solutions defined in \eqref{2.1}.
	
\end{enumerate}
\end{theorem}

\begin{proof}
Using \eqref{4.44} and \eqref{4.47}, we observe that the auxiliary scalar quantity $P(x)$ defined in 
\eqref{4.53} is expressed in terms of the potential pair $(q,r)$ as
\begin{equation}
\label{4.60}
P(x)=\displaystyle\frac{q(x)\,r(x)}{4}.
\end{equation}
Hence, with the help of \eqref{1.13}, \eqref{4.53}, and \eqref{4.60} we obtain the first equality of \eqref{4.52}. Using \eqref{4.60} in \eqref{1.15}, we get the second equality of \eqref{4.52}.
Thus, the proof of (a) is complete.
Using \eqref{4.52} in \eqref{4.44}, we recover the potential $q$ as in \eqref{4.54}.
Similarly, using \eqref{4.52} in \eqref{4.47}, we recover the potential $r$ as in \eqref{4.55}.
Hence, the proof of (b) is also complete.
We obtain the expressions for the components of the Jost solutions $\psi(\zeta,x)$ and $\bar\psi(\zeta,x)$ listed in \eqref{4.56}--\eqref{4.59}, by using \eqref{4.52} in \eqref{4.25}--\eqref{4.28}, respectively.
\end{proof}

\section{Solutions to the Chen--Lee--Liu system}
\label{section5}

In Section~\ref{section4} we have presented the solution to the inverse scattering problem for \eqref{1.6} by using the Marchenko
method. This has been done by using the data set $\{R(\zeta),\bar R(\zeta),(A,B,C), (\bar A,\bar B,\bar C)\}$ as input to
the Marchenko system \eqref{4.40} and by recovering the potential pair $(q,r)$ from the solution $\mathcal K(x,y)$ 
to \eqref{4.40}, as described in \eqref{4.54} and \eqref{4.55}
with the quantity $P(x)$ expressed as in \eqref{4.53}.
If we use the time-evolved version of our input data set, via the Marchenko method we recover the time-evolved potential pair $(q,r).$ From the inverse scattering method \cite{GGKM1967,L1968} we know that the time-evolved potentials
$q(x,t)$ and $r(x,t)$ form a solution to
the integrable nonlinear system \eqref{1.1}.
In this section, since we deal with the time-evolved quantities, we show the explicit $t$-dependence in our notation for the potentials, the scattering coefficients, and the solution to the Marchenko system.

In the next theorem, we provide the time evolution of the scattering coefficients for \eqref{1.6} and the matrices $C$ and
$\bar C$ appearing in the matrix triplets $(A,B,C)$ and $(\bar A,\bar B,\bar C)$ used to describe the bound-state information.  

\begin{theorem}
\label{theorem5.1}
Assume that the potentials $q$ and $r$ appearing in \eqref{1.6} belong to the Schwartz class $\mathcal S(\mathbb R)$
in $x\in\mathbb R$ for each fixed $t\in\mathbb R.$ Let the time evolution of the potentials $q$ and $r$ be governed by the AKNS pair matrix $\mathcal T$
in \eqref{1.5}. We have the following:

\begin{enumerate}

\item[\text{\rm(a)}] The four reflection coefficients $R(\zeta,t),$ $\bar R(\zeta,t),$ $L(\zeta,t),$ $\bar L(\zeta,t)$ for \eqref{1.6} evolve in the time variable $t$ as
\begin{equation}\label{5.1}
R(\zeta,t)=R(\zeta,0)\, e^{4 i \lambda^2 t},\quad 
\bar R(\zeta,t)=\bar R(\zeta,0)\, e^{-4 i \lambda^2 t}, \qquad \lambda\in\mathbb R,\quad
t\in\mathbb R,
\end{equation}
\begin{equation*}
L(\zeta,t)=L(\zeta,0)\, e^{-4 i \lambda^2 t},\quad 
\bar L(\zeta,t)=\bar L(\zeta,0)\, e^{4 i \lambda^2 t}, \qquad \lambda\in\mathbb R,\quad t\in\mathbb R,
\end{equation*}
where we recall that the parameter $\lambda$ is related to the spectral parameter $\zeta$ as in \eqref{2.7}.

\item[\text{\rm(b)}] 
The four transmission coefficients $T_{\text{\rm{l}}}(\zeta,t),$ $T_{\text{\rm{r}}}(\zeta,t),$ $\bar{T_{\text{\rm{l}}}}(\zeta,t),$
$\bar{T_{\text{\rm{r}}}}(\zeta,t)$ do not change in time, and hence we have
\begin{equation*}
T_{\text{\rm{l}}}(\zeta,t)=T_{\text{\rm{l}}}(\zeta,0),\quad 
T_{\text{\rm{r}}}(\zeta,t)=T_{\text{\rm{r}}}(\zeta,0),\quad 
\bar{T_{\text{\rm{l}}}}(\zeta,t)=\bar{T_{\text{\rm{l}}}}(\zeta,0),\quad 
\bar{T_{\text{\rm{r}}}}(\zeta,t)=\bar{T_{\text{\rm{r}}}}(\zeta,0).
\end{equation*}
Since the transmission coefficients do not change in time, for notational simplicity we use $T_{\text{\rm{l}}}(\zeta),$
$T_{\text{\rm{r}}}(\zeta),$ $\bar{T_{\text{\rm{l}}}}(\zeta),$ and $\bar{T_{\text{\rm{r}}}}(\zeta)$ for 
$T_{\text{\rm{l}}}(\zeta,t),$
$T_{\text{\rm{r}}}(\zeta,t),$ $\bar{T_{\text{\rm{l}}}}(\zeta,t),$ and $\bar{T_{\text{\rm{r}}}}(\zeta,t),$ respectively.

\item[\text{\rm(c)}] The four matrices $A,$ $\bar A,$ $B,$ $\bar B$ appearing in \eqref{3.6} and \eqref{3.7} do not change in time. On the other hand, the matrices $C$ and $\bar C$ appearing in \eqref{3.8} evolve in time as 
\begin{equation}
\label{5.4}
C\mapsto C e^{4i A^2 t},\quad \bar C\mapsto \bar C e^{-4i \bar A^2 t},
\end{equation}
where, for notational simplicity, we use $C$ and $\bar C$ to denote their values at $t=0.$ 
\end{enumerate}
\end{theorem}

\begin{proof} The corresponding proof for \eqref{1.11} can be found in Theorem~2.1 of \cite{AEU2023b}. The proof for \eqref{1.6} is similarly obtained.
\end{proof}

With the help of \eqref{5.1}, we obtain the time evolution of $\hat R(y,t)$ and $\hat{\bar R}(y,t)$ defined in \eqref{4.2} as
\begin{equation}
\label{5.5}
\hat R(y,t)=\displaystyle\frac{1}{2\pi}\displaystyle\int_{-\infty}^\infty  
d\lambda\,\displaystyle\frac{R(\zeta,0)}{\zeta}\,e^{4i\lambda^2 t}e^{i\lambda y},\quad 
\hat{\bar R}(y,t)=\displaystyle\frac{1}{2\pi}
\displaystyle\int_{-\infty}^\infty  d\lambda\,\displaystyle\frac{\bar R(\zeta,0)}{\zeta}\,e^{-4i\lambda^2 t}e^{-i\lambda y}.
\end{equation}
Next, using \eqref{5.4} and \eqref{5.5} in \eqref{4.38} we obtain the time-evolved Marchenko kernels $\Omega(y,t)$ and $\bar\Omega(y,t)$ as
\begin{equation}\label{5.6}
\Omega(y,t)=\hat R(y,t)+C\,e^{4iA^2 t}e^{iAy}B,\quad \bar\Omega(y,t)=\hat{\bar R}(y,t)+\bar C\,e^{-4i\bar A^2 t}e^{-i\bar A y}\bar B.
\end{equation}
We remark on the symmetry and elegance in the connection between \eqref{5.5} and \eqref{5.6}. The scalar term
$e^{4i\lambda^2 t}\,e^{i\lambda y}$ appearing in the integrand in the first equality of \eqref{5.5} suggests the appearance of the matrix exponential
$e^{4iA^2 t}\,e^{iAy}$ in the first equality of \eqref{5.6}. When there is a simple bound state at $\lambda=\lambda_j$ of multiplicity 1, the eigenvalue
of $A$ and the matrix $A$ coincide. Thus, the scalar $\lambda$ appearing in $e^{4i\lambda^2 t}e^{i\lambda y}$ in the first equality of \eqref{5.5} is naturally
replaced by the matrix $A$ in $e^{4iA^2 t}e^{iA y}$ in the first equality of \eqref{5.6}. In a similar manner, the scalar term $e^{-4i\lambda^2 t}e^{-i\lambda y}$
appearing in the integrand in the second equality of \eqref{5.5} suggests the appearance of the matrix exponential $e^{-4i\bar A^2 t}\,e^{-i\bar Ay}$
in the second equality of \eqref{5.6}. When there is a simple bound state at $\lambda=\bar\lambda_j$ of multiplicity 1, the eigenvalue of $\bar A$
and the matrix $\bar A$ coincide. Thus, the scalar $\lambda$ appearing in $e^{-4i\lambda^2 t}e^{-i\lambda y}$ in the second equality of \eqref{5.5} is 
naturally replaced by the matrix $\bar A$ in $e^{-4i\bar A^2 t}\,e^{-i\bar A y}.$

In the reflectionless case, i.e. when the reflection coefficients for \eqref{1.6} are all zero, there are various restrictions on the bound-state
information presented in \eqref{3.1}. Hence, in the reflectionless case, there are restrictions on the matrix triplets $(A,B,C)$ and
$(\bar A,\bar B,\bar C)$ appearing in \eqref{3.6}--\eqref{3.8}. We refer the reader to Theorem~4.3 of \cite{AEU2023b} for the restrictions on the 
bound-state information for \eqref{1.11} in the reflectionless case. Since the transmission coefficients for \eqref{1.6} and for \eqref{1.11} are related to 
each other as in \eqref{2.39} and \eqref{2.40}, the restrictions on the bound-state information for \eqref{1.6} in the reflectionless case can be 
obtained from the corresponding restrictions for \eqref{1.11}. For example, in the reflectionless case, we must have $\mathcal N=\bar{\mathcal N},$ where
$\mathcal N$ and $\bar{\mathcal N}$ are the nonnegative integer quantities defined in \eqref{3.9}. In the reflectionless case, with the help of (4.40) of \cite{AEU2023b} and \eqref{2.40}, we obtain the restriction that the transmission coefficients $T_{\text{\rm{r}}}(\zeta)$ and
$\bar T_{\text{\rm{r}}}(\zeta)$ for \eqref{1.6} have to be related to each other as
\begin{equation}\label{5.7}
T_{\text{\rm{r}}}(\zeta)\,\bar T_{\text{\rm{r}}}(\zeta)=e^{-i\mu/2}, \qquad \lambda\in\mathbb R,
\end{equation}
where we recall that $\lambda$ and $\zeta$ are related to each other as in \eqref{2.7} and $\mu$ is the constant appearing in \eqref{1.15}. In that case, the transmission coefficients 
$T_{\text{\rm{r}}}(\zeta)$ and $\bar T_{\text{\rm{r}}}(\zeta)$ each have meromorphic extensions to the entire complex $\lambda$-plane, and hence
the bound-state $\lambda_j$-values and $\bar\lambda_j$-values become restricted. Such restrictions must be taken into account in considering
the analysis of solutions to the linear system \eqref{1.6} and the nonlinear system \eqref{1.1} in the reflectionless case. 

In the reflectionless case, from \eqref{5.4} we obtain 
\begin{equation}\label{5.8}
\Omega(y,t)=C\,e^{4iA^2 t}e^{iAy}B,\quad \bar\Omega(y,t)=\bar C\,e^{-4i\bar A^2 t}e^{-i\bar A y}\bar B.
\end{equation}
In the next theorem, in the reflectionless case without imposing any of the restrictions on the matrix triplets, we obtain the solution to the Marchenko
system \eqref{4.40} in terms of the matrix triplets $(A,B,C)$ and $(\bar A,\bar B,\bar C)$ with the help of the input $\{\Omega(y,t),\bar\Omega(y,t)\}$ given in \eqref{5.6}.

\begin{theorem}
\label{theorem5.2}
Assume that the potential pair $(q,r)$ appearing in \eqref{1.6} belongs to the Schwartz class $\mathcal S(\mathbb R)$
in $x\in\mathbb R$ for each fixed $t\in\mathbb R,$ where the time evolution of the potential pair is governed by the AKNS pair matrix $\mathcal T$ given in \eqref{1.5}.
Then, in the reflectionless case, the solution $\mathcal K(x,y,t)$ to the Marchenko system \eqref{4.40} is given by

\begin{equation}\label{5.9}
K_1(x,y,t)
=-\bar C\,e^{-i\bar A x}\,\bar\Gamma(x,t)^{-1}\,e^{-i\bar A y-4i\bar A^2 t}\,\bar B,
\end{equation}
\begin{equation}
\label{5.10}
\bar K_1(x,y,t)=\bar C\,e^{-i\bar A x}\,\bar\Gamma(x,t)^{-1}\,e^{-i\bar A x-4i\bar A^2 t}\bar M A\,e^{iA(x+y)+4iA^2t}\,B,
\end{equation}
\begin{equation}
\label{5.11}
K_2(x,y,t)
=C\,e^{iAx}\,\Gamma(x,t)^{-1}\,e^{iAx+4iA^2t} M \bar A e^{-i\bar A (x+y)-4i\bar A^2 t}\,\bar B,
\end{equation}
\begin{equation}\label{5.12}
\bar K_2(x,y,t)=-C\,e^{iAx}\,\Gamma(x,t)^{-1}\,e^{iAy+4iA^2t}\,B,
\end{equation}
where $K_1(x,y,t),$ $K_2(x,y,t),$ $\bar K_1(x,y,t),$ $\bar K_2(x,y,t)$ are the entries of $\mathcal K(x,y,t)$ appearing in \eqref{4.14}.
The $2\times 2$ matrix-valued functions $\Gamma(x,t)$ and $\bar\Gamma(x,t)$ appearing \eqref{5.9}--\eqref{5.12} are expressed in terms of the matrix triplets $(A,B,C)$ and $(\bar A,\bar B,\bar C)$ as
\begin{equation}\label{5.13}
\Gamma(x,t):=I-e^{iAx+4iA^2 t}M\bar Ae^{-2i\bar A x-4i\bar A^2 t}\bar M e^{iAx}, 
\end{equation}
\begin{equation}\label{5.14}
\bar\Gamma(x,t):=I-e^{-i\bar A x-4i\bar A^2 t}\bar MAe^{2iAx+4iA^2 t}Me^{-i\bar A x},
\end{equation}
with $I$ denoting the $2\times 2$ identity matrix and the quantities $M$ and $\bar M$ being the $2\times 2$ matrix-valued constants expressed in terms of the matrix triplets $(A,B,C)$ and $(\bar A,\bar B,\bar C)$ as
\begin{equation}\label{5.15}	
M:=\int_{0}^\infty dz\,e^{iAz}\,B\,\bar C\,e^{-i\bar A z},\quad \bar M:=\int_{0}^\infty dz\,e^{-i\bar A z}\,\bar B\,C\,e^{i A z}.
\end{equation}
\end{theorem}

\begin{proof}
From \eqref{5.8}, by taking the $y$-derivatives we obtain
\begin{equation}\label{5.16}
\Omega'(y,t)=iCAe^{4iA^2t}\,e^{iAy}B, \quad \bar\Omega'(y,t)=-i\bar C\bar A e^{-4i\bar A^2t}\,e^{-i\bar A y}\bar B,
\end{equation}
where the prime is used to denote the $y$-derivative. Using \eqref{5.8} and \eqref{5.16} in the uncoupled system we recover
$K_1(x,y,t)$ and $\bar K_2(x,y,t)$ listed in \eqref{5.9} and \eqref{5.12}, respectively. We then use \eqref{5.9} in the integrand in the first 
line of \eqref{4.43} and use \eqref{5.12} in the integrand in the second line of \eqref{4.43} and recover $\bar K_1(x,y,t)$ and
$K_2(x,y,t)$ given in \eqref{5.10} and \eqref{5.11}, respectively. Since the process is similar to the proof of Theorem~ 5.4 of 
\cite{AEU2023a}, we omit the details. We remark that since the eigenvalues of $A$ are all located in the upper-half complex $\lambda$-plane and the eigenvalues of $\bar A$ are all located in the lower-half complex $\lambda$-plane, the two integrals in \eqref{5.15} both exist.
\end{proof}

In the next theorem, we recover the potential pair $(q,r)$ appearing in \eqref{1.6} in terms of the matrix triplets
$(A,B,C)$ and $(\bar A,\bar B,\bar C)$ in the reflectionless case.

\begin{theorem}
\label{theorem5.3}
Assume that the potential pair $(q,r)$ appearing in \eqref{1.6} belong to the Schwartz class $\mathcal S(\mathbb R)$ in $x\in\mathbb R$ 
for each fixed $t\in\mathbb R,$ where the time evolution of the potential pair is governed by the matrix $\mathcal T$ in \eqref{1.5}. Let the quantities $\Omega(y,t)$ and $\bar\Omega(y,t)$ be the time-evolved reflectionless Marchenko 
kernels defined in \eqref{5.8}. Then, we have the following:

\begin{enumerate}

\item[\text{\rm(a)}] The corresponding key quantity $E(x,t)$ defined in \eqref{1.13} can explicitly be written 
in terms of the matrix triplets $(A,B,C)$ and $(\bar A,\bar B,\bar C)$ as
\begin{equation}\label{5.17}
E(x,t)=\exp\left(2i \displaystyle\int_{-\infty}^x dz\,P(z,t)
\right),
\end{equation}
where $P(x,t)$ is the scalar-valued function of $x$ and $t$
defined in \eqref{4.53} with $\bar K_1(x,x,t)$ and 
$K_2(x,x,t)$ there explicitly expressed in terms of the matrix triplets as
\begin{equation}\label{5.18}
K_1(x,x,t)
=-\bar C\,e^{-i\bar A x}\,\bar\Gamma(x,t)^{-1}\,e^{-i\bar A x-4i\bar A^2 t}\,\bar B,
\end{equation}
\begin{equation}\label{5.19}
\bar K_2(x,x,t)=-C\,e^{iAx}\,\Gamma(x,t)^{-1}\,e^{iAx+4iA^2t}\,B.
\end{equation}
The $2\times 2$ matrix-valued quantity $\Gamma(x,t)$ in \eqref{5.13} is explicitly determined by the matrix triplets $(A,B,C)$ and $(\bar A,\bar B,\bar C)$ as described in \eqref{5.18}. Similarly, the $2\times 2$ matrix-valued quantity $\bar\Gamma(x,t)$ in \eqref{5.14} is explicitly determined by the matrix triplets $(A,B,C)$ and $(\bar A,\bar B,\bar C)$ as described in \eqref{5.19}.

\item[\text{\rm(b)}] 
The corresponding potential pair $(q,r)$ appearing in the linear system \eqref{1.6} can explicitly be expressed in terms of the
matrix triplets $(A,B,C)$ and $(\bar A,\bar B,\bar C)$ as
\begin{equation}\label{5.20}
q(x,t)=\left(2\bar C\,e^{-i\bar A x}\,\bar\Gamma(x,t)^{-1}\,e^{-i\bar A x-4i\bar A^2 t}\,\bar B\right)\exp\left(-2i \displaystyle\int_x^\infty dz\,P(z,t)
\right),
\end{equation}
\begin{equation}\label{5.21}
r(x,t)=\left(2C\,e^{iA x}\,\Gamma(x,t)^{-1}\,e^{iA x+4i A^2 t}\,B\right)\exp\left(2i \displaystyle\int_x^\infty dz\,P(z,t)
\right).
\end{equation}
\end{enumerate}
\end{theorem}

\begin{proof}
We get \eqref{5.18} and \eqref{5.19} from \eqref{5.9} and \eqref{5.12}, respectively, by letting $y\to x^+$ there.
Thus, we obtain \eqref{5.17} from the first equality of \eqref{4.52} with the help of \eqref{4.53} by using \eqref{5.18} and \eqref{5.19} on the right-hand side of \eqref{4.53}. This completes the proof of (a). We obtain \eqref{5.20} from \eqref{4.54} by using \eqref{5.18} on the right-hand side of \eqref{4.54}. Similarly, we obtain \eqref{5.21} by using \eqref{5.19} on the right-hand side of \eqref{4.55}. 
\end{proof}

\section{Explicit examples}
\label{section6}

In this section, we illustrate the solution method for the Chen--Lee--Liu system \eqref{1.1} developed in Section~\ref{section5} with two explicitly solved examples.

In the first example below, we construct a one-soliton solution with multiplicity 1 as a solution to the Chen--Lee--Liu system \eqref{1.1}. This is done by solving the inverse scattering problem for \eqref{1.6} in the reflectionless case via the Marchenko method developed in Section~\ref{section4}. As input to the Marchenko system \eqref{4.40} we use a specific pair of matrix triplets, where we choose the size of each of the six matrices in the triplets as $1\times 1.$ For notational simplicity in our example, we write a $1\times 1$ matrix as a scalar quantity.

\begin{example}
\label{example6.1}
\normalfont
In the reflectionless case, we use the matrix triplets $(A,B,C)$ and $(\bar A,\bar B,\bar C)$ given by
\begin{equation}\label{6.1}
A=\begin{bmatrix}
i
\end{bmatrix},\quad B=\begin{bmatrix}
1
\end{bmatrix},\quad C=\begin{bmatrix}
i
\end{bmatrix},
\quad 
\bar A =\begin{bmatrix}
-i
\end{bmatrix},\quad \bar B=\begin{bmatrix}
1
\end{bmatrix},\quad\bar C=\begin{bmatrix}
-i
\end{bmatrix}.
\end{equation}
Using \eqref{6.1} in \eqref{5.15}, we construct the $1\times 1$ constant matrices $M$ and $\bar M$ as
\begin{equation}\label{6.2}
M=-\displaystyle\frac{i}{2}, \quad \bar M=\displaystyle\frac{i}{2}.
\end{equation}
Next, using \eqref{6.1} and \eqref{6.2} in \eqref{5.13} and \eqref{5.14}, we obtain the $1\times 1$ matrices
$\Gamma(x,t)$ and $\bar\Gamma(x,t)$ as
\begin{equation}\label{6.3}
\Gamma(x,t)=1+\displaystyle\frac{ie^{-4x}}{4}, \quad \bar\Gamma(x)=1-\displaystyle\frac{ie^{-4x}}{4}.
\end{equation}
Then, we use \eqref{6.1} and \eqref{6.3} in \eqref{5.18} and \eqref{5.19} and construct the scalar quantities $K_1(x,x,t)$ and $\bar K_2(x,x,t),$ respectively, and we get 
\begin{equation}\label{6.4}
K_1(x,x,t)=-\displaystyle\frac{4e^{2x+4it}}{1+4ie^{4x}}, \quad \bar K_2(x,x,t)=-\displaystyle\frac{4e^{2x-4it}}{1-4ie^{4x}}.
\end{equation}
Next, using \eqref{6.4} in \eqref{4.53} we obtain the quantity $P(x,t)$ as
\begin{equation}\label{6.5}
P(x,t)=\displaystyle\frac{16e^{4x}}{1+16e^{8x}}.
\end{equation}
Finally, we use \eqref{6.4} and \eqref{6.5} in \eqref{4.52}, \eqref{4.54}, \eqref{4.55}, and we recover the quantity $E(x,t),$ the constant $\mu,$ and the potentials $q(x,t)$ and $r(x,t)$ as
\begin{equation}
\label{6.6}
E(x,t)=e^{2i\,\tan^{-1}\left(4e^{4x}\right)}, \quad \mu=2\pi,
\end{equation}
\begin{equation}\label{6.7}
q(x,t)=\displaystyle\frac{8e^{2x+4it+2i\,\tan^{-1}\left(4e^{4x}\right)}}{1+4ie^{4x}}, \quad
r(x,t)=\displaystyle\frac{8e^{2x-4it-2i\,\tan^{-1}\left(4e^{4x}\right)}}{1-4ie^{4x}}, \qquad x\in\mathbb R, \quad t\in \mathbb R.
\end{equation}
From \eqref{6.7} we observe that the potentials $q(x,t)$ and $r(x,t)$ satisfy 
\begin{equation*}
r(x,t)=q(x,t)^\ast,
\end{equation*}
where we use an asterisk to denote complex conjugation.
From \eqref{6.7} we also see that the potentials $q(x,t)$ and $r(x,t)$ have no singularities and
they belong to the Schwartz class $\mathcal S(\mathbb R)$ in $x$ for each fixed $t\in\mathbb R.$
Since the $1\times 1$ matrices $A$ and $\bar A$ in the input data set \eqref{6.1} correspond to the poles of $T_{\text{\rm{r}}}(\zeta)$
and $\bar T_{\text{\rm{r}}}(\zeta)$ in the upper-half and lower-half complex $\lambda$-planes, respectively, with the help of \eqref{5.7} we determine those two right transmission coefficients as
\begin{equation*}
T_{\text{\rm{r}}}(\zeta)=-\displaystyle\frac{\lambda+i}{\lambda-i}, \quad
\bar{T_{\text{\rm{r}}}}(\zeta)=\displaystyle\frac{\lambda-i}{\lambda+i}, \qquad \lambda\in\mathbb C,
\end{equation*}
where the value of $\mu$ appearing in the second equality in \eqref{6.6} is taken into account and we recall that the parameter $\lambda$ is related to the spectral parameter $\zeta$ as in \eqref{2.7}.
In this example, as seen from the first equality of \eqref{6.6}, the quantity $E(x,t)$ is independent of $t$ even though the potentials $q(x,t)$ and $r(x,t)$ given in \eqref{6.7} contain the parameter $t.$ On the other hand, the quantities   
$|q(x,t)|$ and $|r(x,t)|$ do not change in $t.$ Since $q(x,t)$ and $r(x,t)$ are complex valued, in Figure~\ref{figure6.1} we have plotted $|q(x,t)|$ and $|r(x,t)|$ as functions of $x.$ From \eqref{6.7} we have $|r(x,t)|=|q(x,t)|$ for all $x\in\mathbb R$ and $t\in\mathbb R.$ Hence, in this example, the soliton represented by $|q(x,t)|$ or $|r(x,t)|$ does not move in time.

\begin{figure}[!h]
     \centering
         \includegraphics[width=2in]{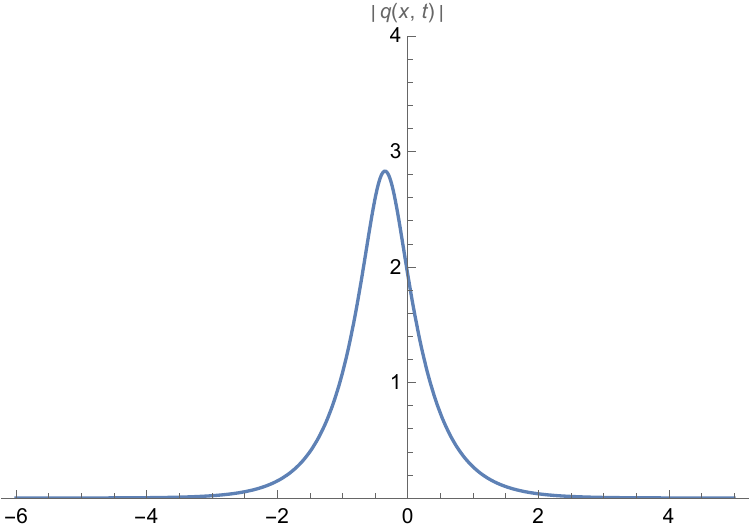}      \hskip .6in
         \includegraphics[width=2in]{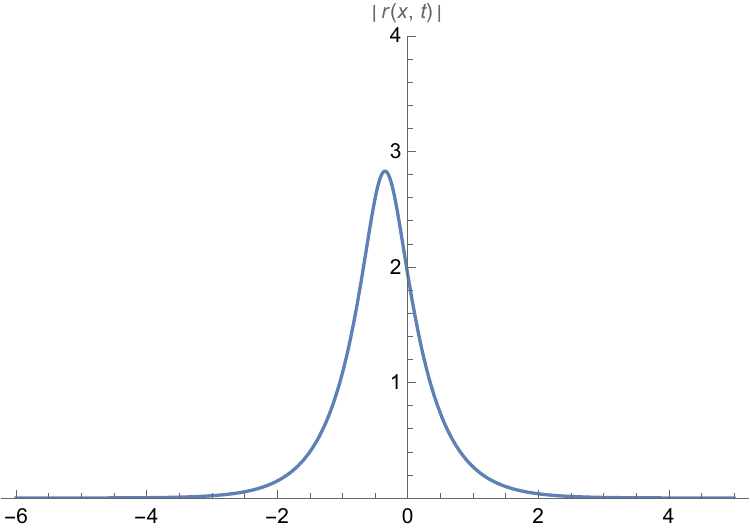}
\caption{The snapshots for $|q(x,t)|$ and $|r(x,t)|$ in Example~\ref{example6.1} corresponding to the input data set in \eqref{6.1}.}
\label{figure6.1}
\end{figure}
\end{example}

In the next example, we illustrate a two-soliton solution to \eqref{1.1}. As in Example~\ref{example6.1} this is done by solving the inverse scattering problem for \eqref{1.6} in the reflectionless case via the Marchenko method. As input to the Marchenko system \eqref{4.40}, we use a pair of matrix triplets with the $2\times 2$ matrices $A$ and $\bar A$ presented in their Jordan canonical forms.

\begin{example}
\label{example6.2}
\normalfont
In the reflectionless case, we choose the input data set to the Marchenko system as
\begin{equation}\label{6.10}
A =\begin{bmatrix}
i&1\\ \noalign{\medskip}
0&i
\end{bmatrix},\quad B=\begin{bmatrix}
0\\ \noalign{\medskip}
1
\end{bmatrix},\quad C=\begin{bmatrix}
i&1
\end{bmatrix},
\quad
\bar A =\begin{bmatrix}
-i&1\\ \noalign{\medskip}
0&-i
\end{bmatrix},\quad \bar B=\begin{bmatrix}
0\\ \noalign{\medskip}
1
\end{bmatrix},\quad\bar{ C}=\begin{bmatrix}
-i&1
\end{bmatrix}.
\end{equation}
We recover the potential pair $(q,r)$ in terms of the matrix triplets
$(A,B,C)$ and $(\bar A,\bar B,\bar C)$ as follows. Using \eqref{6.10} in \eqref{5.15}, we construct the $2\times 2$
constant matrices $M$ and $\bar M$ as
\begin{equation}\label{6.11}
M=\begin{bmatrix}
\displaystyle\frac{1}{4}&0\\ \noalign{\medskip}
-\displaystyle\frac{i}{2}&\displaystyle\frac{1}{4}
\end{bmatrix}, \quad \bar M=\begin{bmatrix}
\displaystyle\frac{1}{4}&0\\ \noalign{\medskip}
\displaystyle\frac{i}{2}&\displaystyle\frac{1}{4}
\end{bmatrix}.
\end{equation}
Then, using \eqref{6.10} and \eqref{6.11} in \eqref{5.13} and \eqref{5.14}, we obtain the $2\times 2$
matrices $\Gamma(x,t)$ and $\bar\Gamma(x,t),$ respectively, as
\begin{equation}\label{6.12}
\Gamma(x)=\begin{bmatrix}
\Gamma_{11}&\Gamma_{12}\\ 
\noalign{\medskip}
\Gamma_{21}&\Gamma_{22}
\end{bmatrix}, \quad \bar\Gamma(x)=\begin{bmatrix}
\Gamma_{11}^\ast&\Gamma_{12}^\ast\\ 
\noalign{\medskip}
\Gamma_{21}^\ast&\Gamma_{22}^\ast
\end{bmatrix},
\end{equation}
where we have defined
\begin{equation*}
\Gamma_{11}:=1+\displaystyle\frac{e^{-4x}}{16}\left(-i+8\left(ix^2-4xt+6t+32it^2\right)\right),
\end{equation*}
\begin{equation*}
\Gamma_{12}:=\displaystyle\frac{e^{-4x}}{16}\left(-1+4x^2-8x^3-128t^2(-1+2x)-32it(x-1^2)\right),
\end{equation*}
\begin{equation*}
\Gamma_{21}:=\displaystyle\frac{e^{-4x}}{4}\left(-1+2x-8it\right),
\end{equation*}
\begin{equation*}
\Gamma_{22}:=1+\displaystyle\frac{e^{-4x}}{16}\left(i(3-8x+8x^2)+16t(-1+2x)\right),
\end{equation*}
and we recall that an asterisk denotes complex conjugation.
Having constructed the matrices $\Gamma(x,t)$ and $\bar\Gamma(x,t),$ we use \eqref{6.10} and \eqref{6.12} in
\eqref{5.18} and \eqref{5.19}, and we obtain $K_1(x,x,t)$ and $\bar K_2(x,x,t),$ respectively, as
\begin{equation}\label{6.17}
K_1(x,x,t)=\displaystyle\frac{32e^{2x+4it}\left(-x+8e^{4x}\left(2ix-i+8t\right)-4it\right)}{32e^{4x}\left(1-4x+8x^2+128t^2+16it\right)-i+256ie^{8x}}, 
\end{equation}
\begin{equation}\label{6.18}
\bar K_2(x,x,t)=\displaystyle\frac{32e^{2x-4it}\left(-x+8e^{4x}\left(-2ix+i+8t\right)+4it\right)}{32e^{4x}\left(1-4x+8x^2+128t^2-16it\right)+i-256ie^{8x}}.
\end{equation}
Next, we use \eqref{6.17} and \eqref{6.18} in \eqref{4.53}, and we get the quantity $P(x,t)$ as
\begin{equation}\label{6.19}
P(x,t)=\displaystyle\frac{1024\,e^{4x}\,w_1\,w_1^\ast}{w_2\,w_2^\ast},
\end{equation}
where we have defined 
\begin{equation*}
w_1:=ix-4\,t-8i\,e^{4x}(-i+8t+2i\,x),
\end{equation*}
\begin{equation*}
w_2:=-1+256\,e^{8x}-32ie^{4x}\left(1-4x+8x^2+16it+128t^2\right).
\end{equation*}
From \eqref{6.19} we observe that the quantity $P(x,t)$ is real valued for all $x\in\mathbb R$ and $t\in\mathbb R.$ Finally, using \eqref{6.17}, \eqref{6.18}, and \eqref{6.19} in \eqref{4.52}, \eqref{4.54}, and \eqref{4.55}, we recover the quantity $E(x,t),$ the constant $\mu,$ and the potentials $q(x,t)$ and $r(x,t)$ as
\begin{equation}
\label{6.22}
E(x,t)=\exp \bigg(-2i \tan^{-1}(w_3/w_4)\bigg), \quad \mu=4\pi,
\end{equation}
\begin{equation}\label{6.23}
q(x,t)=-\displaystyle\frac{64\,w_1}{w_2} \exp\bigg(2x+4i\,t-2i\,\tan^{-1}(w_3/w_4)\bigg),  \qquad x\in\mathbb R, \quad t\in \mathbb R,
\end{equation}
\begin{equation}\label{6.24}
 \quad r(x,t)=q(x,t)^*,  \qquad x\in\mathbb R, \quad t\in \mathbb R,
\end{equation}
where we have defined
\begin{equation*}
w_3:=32e^{4x}\left(1-4x+8x^2+128t^2\right), \quad w_4:=-1+256e^{8x}+512e^{4x}t.
\end{equation*}
From \eqref{6.23} and \eqref{6.24} we observe that the potentials $q(x,t)$ and $r(x,t)$ satisfy $|r(x,t)|=|q(x,t)|$ and they each belong to the Schwartz class
$\mathcal S(\mathbb R)$ in $x$ for $t\in\mathbb R.$
We also see from \eqref{6.22} that the quantity $E(x,t)$ depends on $t,$ whereas in Example~\ref{example6.1} the quantity
$E(x,t)$ given in \eqref{6.6} does not depend on $t.$ The eigenvalues of the $2\times 2$ matrices $A$ and $\bar A$ correspond to the poles of $T_{\text{\rm{r}}}(\zeta)$
and $\bar T_{\text{\rm{r}}}(\zeta)$ in the upper-half and lower-half complex $\lambda$-planes, respectively. Hence, with the help of \eqref{5.7} we see that those two right transmission coefficients are given by  
\begin{equation}\label{6.26}
T_{\text{\rm{r}}}(\zeta)=\left(\displaystyle\frac{\lambda+i}{\lambda-i}\right)^2,  \quad
\bar{T_{\text{\rm{r}}}}(\zeta)=\left(\displaystyle\frac{\lambda-i}{\lambda+i}\right)^2, \qquad \lambda\in\mathbb C,
\end{equation}
where the value of $\mu$ listed in the second equality of \eqref{6.22} is taken into account
and we recall that the parameters $\lambda$ and $\zeta$ are related to each other as in \eqref{2.7}. As seen from \eqref{6.26}, each transmission coefficient has a double 
pole. Since $q(x,t)$ is complex valued and $|r(x,t)|=|q(x,t)|,$ we only discuss the time evolution of $|q(x,t)|.$ In Figure~\ref{figure6.2} we show the snapshots for $|q(x,t)|$ at $t=-5,$ $t=-0.2,$ $t=0,$ $t=0.1,$ $t=0.2,$ and  $t=5,$ respectively. As seen from 
Figure~\ref{figure6.2}, there are two solitons that are initially far apart. They move toward each other and a Mathematica animation shows that their speeds increase 
as they get closer. Then, the two solitons interact with each other nonlinearly, and then they move away from each other. As they move away 
from each other, they regain their individual shapes.
\begin{figure}[!h]
     \centering
         \includegraphics[width=2in]{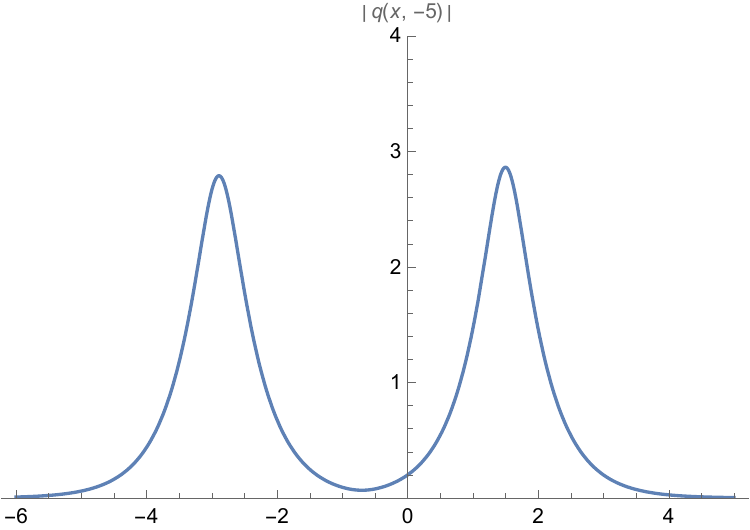}      \hskip .1in
         \includegraphics[width=2in]{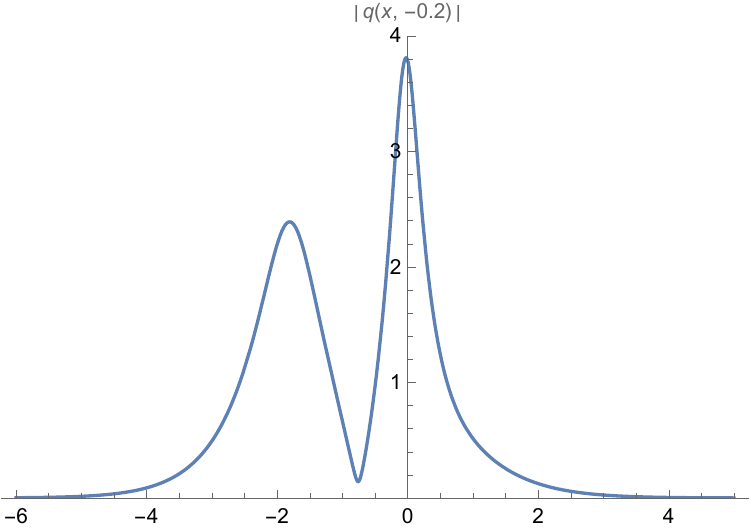} \hskip .1in
         \includegraphics[width=2in]{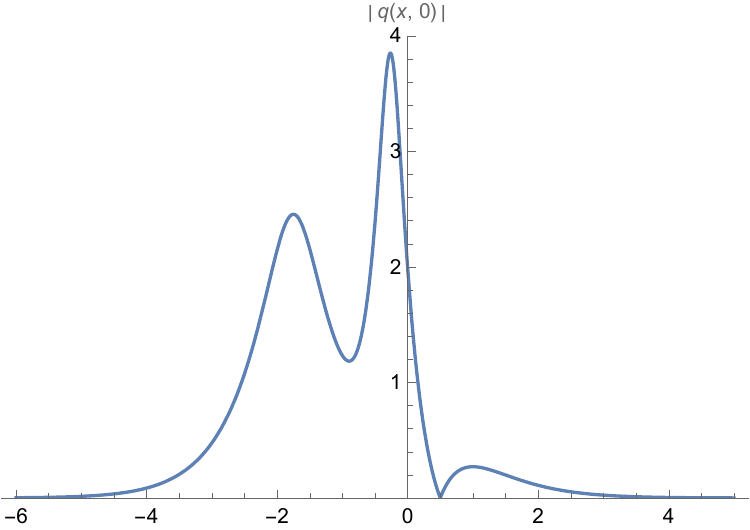}
         \vskip .2in
         \includegraphics[width=2in]{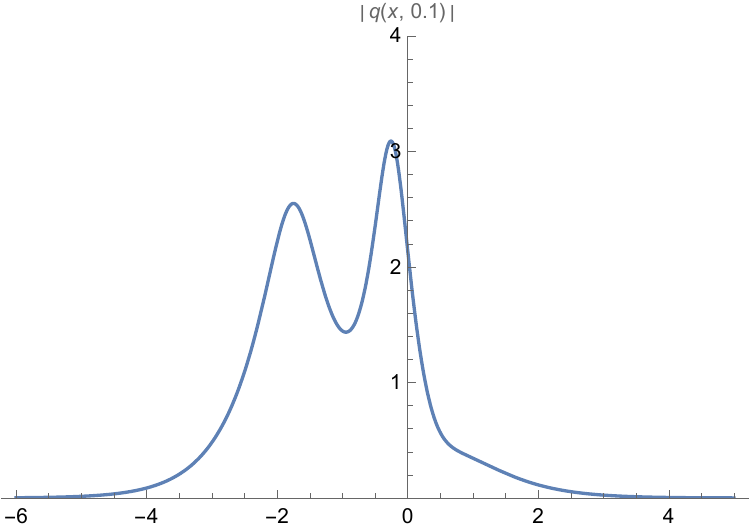} \hskip .1in
         \includegraphics[width=2in]{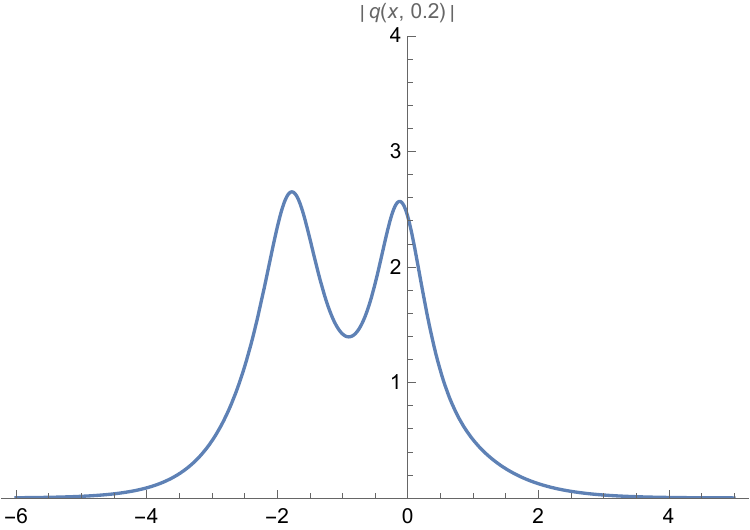} \hskip .1in
         \includegraphics[width=2in]{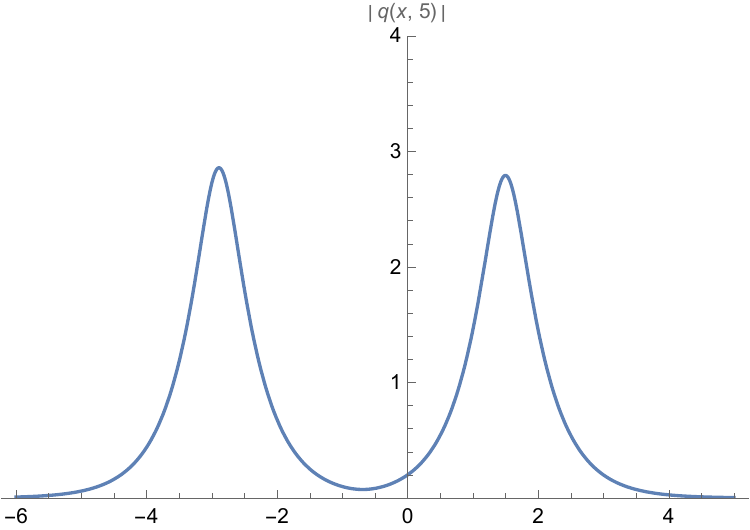}
\caption{The snapshots for $|q(x,t)|$ in Example~\ref{example6.2} at $t=-5,$ $t=-0.2,$ $t=0,$ $t=0.1,$ $t=0.2,$ and  $t=5,$ respectively.}
\label{figure6.2}
\end{figure}
\end{example}

In the reflectionless case, we have prepared a Mathematica notebook that allows the user to input the entries of the two matrix triplets corresponding to any number of bound states with any multiplicities. By using the method of Section~\ref{section5}, our Mathematica notebook evaluates all the relevant quantities and yields the output including the scalar quantity $E(x,t)$ in \eqref{1.11}, the constant $\mu$
in \eqref{1.13}, the potentials $q(x,t)$ and $r(x,t)$ appearing in \eqref{1.1} and \eqref{1.6}, and the transmission coefficients $T_{\text{\rm{r}}}(\zeta)$ and $\bar T_{\text{\rm{r}}}(\zeta)$ appearing in \eqref{5.7}. Our Mathematica notebook also allows the user to animate $|q(x,t)|$ and $|r(x,t)|$ in order to observe the time evolution of soliton solutions to \eqref{1.1} corresponding to any number of bound states and any number of multiplicities. 

The advantage of using matrix exponentials in expressing explicit solutions to \eqref{1.1} becomes clear as the number of bound states or their multiplicities become large. As seen from \eqref{5.20} and \eqref{5.21}, the potentials $q(x,t)$ and $r(x,t)$ are expressed in a compact form with the help of matrix exponentials constructed by using a pair of matrix triplets. By ``unpacking'' those matrix exponentials, we can express $q(x,t)$ and $r(x,t)$
in terms of elementary functions, where those latter expressions become extremely lengthy as the number of bound states and their multiplicities increase, whereas
the compact expressions in \eqref{5.20} and \eqref{5.21} involving matrix exponentials and matrix triplets remain unchanged no matter how many bound states we have and no matter how large their multiplicities are.

\end{document}